\documentclass[final]{dmtcs-episciences}

\usepackage{subfigure}
\usepackage[utf8]{inputenc}
 \usepackage{amsmath,amsfonts,amssymb}
 \usepackage{mathtools}
\usepackage{multirow}

\usepackage[round]{natbib}

\title{New Results on Directed Edge Dominating Set\thanks{Partially supported by JSPS and MAEDI under the Japan-France Integrated Action Program (SAKURA), Project GRAPA 38593YJ. An extended abstract of this paper appeared in \cite{BelmonteHK0L18}.}}

\author[R\'{e}my Belmonte et al.]{
 R\'{e}my Belmonte\affiliationmark{1}
\and Tesshu Hanaka\affiliationmark{2}
\and Ioannis Katsikarelis\affiliationmark{3}
\and Eun Jung Kim\affiliationmark{3}\thanks{Partially supported by ANR grants ``S-EX-AP-PE-AL'' (ANR-21-CE48-0022), ``ASSK'' (ANR-18-CE40-0025-01), and ``ESIGMA" (ANR-17-CE23-0010).}
\and {Michael Lampis\affiliationmark{3}\thanks{Partially supported by ANR grants ``S-EX-AP-PE-AL'' (ANR-21-CE48-0022), ``ASSK'' (ANR-18-CE40-0025-01), and ``ESIGMA" (ANR-17-CE23-0010).}}}

\affiliation{
University of Electro-Communications, Chofu, Tokyo, 182-8585, Japan\\
Department of Information and System Engineering, Chuo University, Tokyo, Japan\\
Universit\'{e} Paris-Dauphine, PSL Research University, CNRS, UMR 7243, LAMSADE, 75016, Paris, France}


\newtheorem{theorem}{Theorem}
\newtheorem{lemma}[theorem]{Lemma}
\newtheorem{definition}[theorem]{Definition}
\newtheorem{fact}[theorem]{Fact}
\newcommand{\tw}{\ensuremath\textrm{tw}}
\newcommand{\pw}{\ensuremath\textrm{pw}}

\newcommand{\vc}{\ensuremath\textrm{vc}}

\newcommand{\eds}{\textsc{EDS}}
\newcommand{\dist}{\ensuremath \textsf{dist}}
\newtheorem{claim}{Claim}[theorem]

\date{}

\begin{document}

\publicationdata{vol. 25:1}{2023}{4}{10.46298/dmtcs.5378}{2019-04-14; 2019-04-14; 2020-09-01; 2022-08-17; 2023-01-04}{2023-01-05}

\maketitle

\begin{abstract}

We study a family of generalizations of \textsc{Edge Dominating Set} on
directed graphs called \textsc{Directed $(p,q)$-Edge Dominating Set}. In this problem an
arc $(u,v)$ is said to dominate itself, as well as all arcs which are at
distance at most $q$ \emph{from} $v$, or at distance at most $p$ \emph{to} $u$.

First, we give significantly improved FPT algorithms for the two most important
cases of the problem, $(0,1)$-d\eds\ and $(1,1)$-d\eds\ (that correspond to
versions of \textsc{Dominating Set} on line graphs), as well as polynomial
kernels. We also improve the best-known approximation for these cases from
logarithmic to constant. In addition, we show that $(p,q)$-d\eds\ is FPT
parameterized by $p+q+\tw$, but W-hard parameterized by $\tw$ (even if the size
of the optimum is added as a second parameter), where $\tw$ is the treewidth of
the underlying (undirected) graph of the input.

We then go on to focus on the complexity of the problem on tournaments. Here,
we provide a complete classification for every possible fixed value of $p,q$,
which shows that the problem exhibits a surprising behavior, including cases
which are in P; cases which are solvable in quasi-polynomial time but not in P;
and a single case $(p=q=1)$ which is NP-hard (under randomized reductions) and
cannot be solved in sub-exponential time, under standard assumptions.

\end{abstract}
\section{Introduction}\label{sec_intro}

\textsc{Edge Dominating Set} (\eds) is a classical graph problem, equivalent to
\textsc{Minimum Dominating Set} on line graphs.  Despite the problem's
prominence, \eds\ has until recently received very little attention in the
context of directed graphs.  In this paper we investigate the complexity of a
family of natural generalizations of this problem to digraphs,
building upon the recent work of \cite{HanakaNO17}.

One of the reasons that \eds\ has not been well-studied so far in digraphs is
that there are several natural ways in which the undirected version can be
generalized. For example, seeing as \eds\ is exactly \textsc{Dominating Set} in
line graphs, one could define \textsc{Directed} \eds\ as \textsc{(Directed)
Dominating Set} in line digraphs, similarly to \cite{Harary60}. In this formulation, an arc
$(u,v)$ dominates all arcs $(v,w)$; however $(v,w)$ does not dominate $(u,v)$.
Another natural way to define the problem would be to consider
\textsc{Dominating Set} on the underlying graph of the line digraph, so as to
maximize the symmetry of the problem, while still taking into account the arcs'
directions. In this formulation, $(u,v)$ dominates arcs coming out of
$v$ and arcs coming into $u$, but not any other arcs incident on $u,v$.

A unifying framework for studying such formulations was recently given by
\cite{HanakaNO17}, that defined $(p,q)$-d\eds\ for any two non-negative
integers $p,q$. In this setting, an arc $(u,v)$ dominates every other arc which
lies on a directed path of length at most $q$ that begins at $v$, or lies on a
directed path of length at most $p$ that ends at $u$. In other words, $(u,v)$
dominates arcs in the forward direction up to distance $q$, and in the backward
direction up to distance $p$.  The interest in defining the problem in such a general
manner is that it allows us to capture at the same time \textsc{Directed
Dominating Set} on line digraphs ($(0,1)$-d\eds), \textsc{Dominating Set} on the
underlying graph of the line digraph ($(1,1)$-d\eds), as well as versions
corresponding to $r$-\textsc{Dominating Set} in the line digraph.  We thus
obtain a family of optimization problems on digraphs, with varying degrees of
symmetry, all of which crucially depend on the directions of arcs in the input
digraph.

\begin{table}[t]
\centering 
\begin{tabular}{c | c | c | c | c} 
\hline\hline 
\small{Param.}  & $p,q$ & FPT / W-hard &  Kernel & Approximability \\ [0.5ex] 
\hline 
\multirow{3}{*}{$k$} 	&	$p+q\leq 1$	
					& 	$2^{O(k)}$ {\tiny \cite{HanakaNO17}} $\rightarrow$ $2^k$ \tiny{[Thm.\ref{thm_2branching}]}
					&     $O(k)$ vertices \tiny{[Thm.\ref{prop:kernel01}]}
					& 	3-approx \tiny{[Thm.\ref{prop:01approx}]} \\ 
		&	$p=q=1$ 	
					& 	$2^{O(k)}$ {\tiny \cite{HanakaNO17}} $\rightarrow$ $9^k$ \tiny{[Thm.\ref{thm_9branching}]}
					&     $O(k^2)$ vertices \tiny{[Thm.\ref{prop:kernel11}]}
					& 	8-approx \tiny{[Thm.\ref{thm_approx}]} \\ 
		&	$\max\{p,q\}\geq 2$
					& 	W[2]-hard~\tiny{\cite{HanakaNO17}}
					&    -
					& 	no $o(\ln k)$-approx~\tiny{\cite{HanakaNO17}}  \\ 
\hline
{\sf tw}		 	&  any $p,q$ & W[1]-hard {\tiny [Thm.\ref{thm:tw_whard}]}		&  - 		&  - \\	
\hline
	{\sf tw}+$p+q$	 	&  any $p,q$ & FPT {\tiny [Thm.\ref{thm_tw_DP}]}		&   unknown 		&  - \\	
\hline 
\end{tabular}
\caption{Complexity status for various values of $p$ and $q$: on general digraphs.} 
\label{table:general} 
\end{table}

\paragraph{Our contribution:} In this paper we advance the state-of-the-art on the complexity of \textsc{Directed $(p,q)$-Edge Dominating Set} on two fronts.\footnote{We
note that in the remainder we always assume that $p\le q$, as in the case where
$p>q$ we can reverse the direction of all arcs and solve $(q,p)$-d\eds.}

First, we study the complexity and approximability of the problem in general (see
Table~\ref{table:general}).
The problem is shown NP-hard for all values of $p,q$ (except $p=q=0$), even for
planar bounded-degree DAGs by \cite{HanakaNO17}, so it makes sense to study its
parameterized complexity and approximability. We show that its two most natural
cases, $(1,1)$-d\eds\ and $(0,1)$-d\eds, admit FPT algorithms with running times
$9^k$ and $2^k$, respectively, where $k$ is the size of the optimal solution.
These algorithms significantly improve upon the FPT algorithms given by
\cite{HanakaNO17}, that use the fact that the treewidth (of the underlying
graph of the input) is at most $2k$ and runs dynamic programming over a
tree decomposition of width at most $10k$, obtained by the algorithm
of~\cite{BodlaenderDDFLP16}. The resulting running-time estimate for the
algorithm of \cite{HanakaNO17} is thus around $25^{10k}$.  Though both of our
algorithms rely on standard branching techniques, we make use of several
non-trivial ideas to obtain reasonable bases in their running times.
We also
show that both of these problems admit polynomial kernels. These are the only
cases of the problem which may admit such kernels, since the problem is shown W-hard
for all other values of $p,q$ by \cite{HanakaNO17}.
Furthermore, we give an $8$-approximation for $(1,1)$-d\eds\ and a
$3$-approximation for $(0,1)$-d\eds.  We recall that \cite{HanakaNO17} showed
an $O(\log n)$-approximation for general values of $p,q$, and a matching
logarithmic lower bound for the case $\max\{p,q\}\ge 2$.  Therefore our result
completes the picture on the approximability of the problem by showing that the
only two currently unclassified cases belong to~APX. 
Finally, we consider the problem's complexity parameterized by the treewidth of
the underlying graph. We show that, even though the problem is FPT when all of
$p,q,\tw$ are parameters, it is in fact W[1]-hard if parameterized only by
$\tw$.  More strongly, we show that the problem is W[1]-hard when
parameterized by the pathwidth and the size of the optimum.

Our second contribution in this paper is an analysis of the complexity of the
problem on tournaments, which are one of the most well-studied classes of
digraphs (see Table~\ref{table:tournament}). One of the reasons for focusing on
this class is that the complexity of \textsc{Dominating Set} has a peculiar
status on tournaments, as it is solvable in quasi-polynomial time, W[2]-hard,
but neither in P nor NP-complete (under standard assumptions).  Here, we
provide a \emph{complete classification} of the problem which paints an even
more surprising picture. We show that $(p,q)$-d\eds\ goes from being in P for
$p+q\le 1$; to being APX-hard and unsolvable in $2^{n^{1-\epsilon}}$ under the
(randomized) ETH for $p=q=1$; to being equivalent to \textsc{Dominating Set} on
tournaments, hence NP-intermediate, quasi-polynomial-time solvable, and
W[2]-hard, when one of $p$ and $q$ equals $2$; and finally to being
polynomial-time solvable again if $\max\{p,q\}\ge 3$ and neither $p$ nor~$q$
equals~2. We find these results surprising, because few problems demonstrate
such erratic complexity behavior when manipulating their parameters and
because, even though in many cases the problem does seem to behave like
\textsc{Dominating Set}, the fact that $(1,1)$-d\eds\ becomes significantly
harder shows that the problem has interesting complexity aspects of its own.
The most technical part of this classification is a reduction that
establishes the hardness of $(1,1)$-d\eds, making use of several
\emph{randomized} tournament constructions, that we show satisfy certain
desirable properties with high probability; as a result our reduction itself is
randomized.  

\begin{table}[ht]

\centering 
\begin{tabular}{c | c} 
\hline\hline 
Range of $p,q$  & Complexity \\ [0.5ex] 
\hline 
$p=q=1$ &  NP-hard {\tiny [Thm.~\ref{thm_11_NP}]}, FPT {\tiny [Thm.~\ref{thm_9branching}]}, polynomial kernel {\tiny [Thm.~\ref{prop:kernel11}]}\\
$p=2$ or $q= 2$ & Quasi-P-time {\tiny [Thm.~\ref{thm_02_12_22_domsetALG}]}, W[2]-hard {\tiny [Thm~\ref{thm_02_12_22_domsetCom}]}   \\
remaining cases  & P-time {\tiny [Thm.~\ref{thm_01_P} and~\ref{thm_pq3_P}]}  \\
\hline 
\end{tabular}
\caption{Complexity status for various values of $p$ and $q$: on tournaments.}
\label{table:tournament} 
\end{table}

\paragraph{Related Work:} On undirected graphs \textsc{Edge Dominating
Set}, also known as \textsc{Maximum Minimal Matching}, is NP-complete even on
bipartite, planar, bounded degree graphs as well as other special cases, see
\cite{Yannakakis80,Horton1993}. It can be approximated within a factor of 2 as shown by
\cite{FujitoN02} (or better in some special cases as shown by
\cite{CardinalLL09,SchmiedV12,Baker94}), but not a factor better than $7/6$
according to \cite{Chlebik2006} unless P=NP. The problem has been the subject of intense
study in the parameterized and exact algorithms community (\cite{XiaoN14}),
producing a series of improved FPT algorithms by
\cite{Fernau06,Binkele-RaibleF10,FominGSS09,XiaoKP13}; the current best is
given by \cite{IwaideN16}. A kernel with $O(k^2)$ vertices and $O(k^3)$ edges
is also shown by \cite{Hagerup12}.

For $(p,q)$-d\eds, \cite{HanakaNO17} show the problem to be NP-complete on
planar DAGs, in P on trees, and W[2]-hard and $c\ln k$-inapproximable on DAGs
if $\max\{p,q\} >1$. The same paper gives FPT algorithms for $\max\{p,q\}\le
1$.  Their algorithm performs DP on a tree decomposition of width $w$ in
$O(25^w)$, using the fact that $w\le 2k$ and the algorithm of
~\cite{BodlaenderDDFLP16} to obtain a decomposition of width $10k$.  

\textsc{Dominating Set} is shown to not admit an $o(\log n)$-approximation by
\cite{DinurS14,Moshkovitz15}, and to be W[2]-hard and unsolvable in time
$n^{o(k)}$ under the ETH by \cite{Downey1995,CyganFKLMPPS15}. The problem is
significantly easier on tournaments, as the size of the optimum is always at most $\log n$,
hence there is a trivial $n^{O(\log n)}$ (quasi-polynomial)-time algorithm. It
remains, however, W[2]-hard as shown by \cite{Downey95}.  The problem thus finds itself in
an intermediate space between P and NP, as it cannot have a polynomial-time
algorithm unless FPT=W[2] and it cannot be NP-complete under the ETH (as it
admits a quasi-polynomial-time algorithm). The generalization of
\textsc{Dominating Set} where vertices dominate their $r$-neighborhood has also
been well-studied in general, e.g.\ by
\cite{BorradaileL16,DemaineFHT05,EisenstatKM14,KatsikarelisLP17,Kreutzer2012}. It is noted by \cite{BiswasJRS22} that this problem is much easier on tournaments for $r\ge 2$, as the size of the solution is always a constant.

\section{Definitions and Preliminaries}\label{sec_defs}

\paragraph{Graphs and domination:} We use standard graph-theoretic
notation. If $G=(V,E)$ is a graph, $S \subseteq V$ a subset of vertices and $A\subseteq E$ a subset of edges,
then $G[S]$ denotes the subgraph of $G$ induced by $S$, while $G[A]$ denotes the subgraph of $G$ that includes $A$ and all its endpoints. We let $V=A\dot{\cup} B$ denote the disjoint set union of $A$ and $B$. For a vertex $v \in V$, the
set of neighbors of $v$ in $G$ is denoted by $N_G(v)$, or simply $N(v)$, and
$N_G(S):=(\bigcup_{v \in S}N(v)) \setminus S$ will be written as
$N(S)$. We define $N[v]:=N(v) \cup \{v\}$ and $N[S]:=N(S) \cup S$.

Depending on the context, we use $(u,v)$ for $u,v \in V$ to denote either an
undirected edge connecting two vertices $u,v$, or an \emph{arc} (a directed
edge) with \emph{tail} $u$ and \emph{head} $v$. 
An \emph{incoming} (resp.\ \emph{outgoing}) arc for vertex $v$ is an arc whose head (resp.\ tail) is $v$. 
In a directed graph $G=(V,E)$, the set of \emph{out-neighbors} (resp.\ \emph{in-neighbors}) of a vertex $v$ is defined as $\{u\in V: (v,u)\in E\}$ (resp. $\{u\in V: (u,v)\in E\}$) and denoted as $N_G^+(v)$ 
(resp. $N_G^-(v)$). Similarly to the case of undirected graphs, $N^+(S)$ and $N^-(S)$ respectively stand for the sets $(\bigcup_{v\in S} N^+(v))\setminus S$ and  $(\bigcup_{v\in S} N^-(v))\setminus S$.
For a subdigraph $H$ of  $G$ and subsets $S,T\subseteq V$, we  let $\delta_H(S,T)$ denote the set of arcs in $H$ whose tails are in $S$ and heads are in $T$.

We use $\delta_H^-(S)$ (resp.\ $\delta_H^+(S)$) to denote the set $\delta_H(V\setminus S,S)$ (resp.\ the set $\delta_H(S,V\setminus S)$). 
If $S$ is a singleton consisting of a vertex $v$, we write $\delta_H^+(v)$ (resp.\  $\delta_H^-(v)$) instead of $\delta_H^+(\{v\})$ (resp.\ $\delta_H^-(\{v\})$). 
The union $\delta_H^+(v)\cup \delta_H^-(v)$ is denoted as $\delta_H(v).$
The \emph{in-degree} $d_H^-(v)$ (respectively \emph{out-degree} $d_H^+(v)$) of a vertex $v$ is defined as $|\delta_H^-(v)|$ (resp.\ $|\delta_H^+(v)|)$), and we write $d_H(v)$ to denote $d^+_H(v)+d^-_H(v)$. We omit $H$ if it is clear from the context. If $H$ is $G[A]$ for some vertex or arc set of $G$, then we write $A$ in place of $G[A]$. 

A \emph{source} (resp.\ \emph{sink}) is a vertex that has no incoming (resp.\ outgoing) arcs.
A vertex $v$ is said to {\em in-cover} every incoming arc $(u, v)$ and {\em out-cover} every outgoing arc $(v, u)$ for some $u$. Here, for a path $v_1, v_2, \ldots, v_l$,  the {\em length} of the path is defined as the number of arcs, that is, $l-1$.

A directed graph is \emph{strongly connected} if there is a path in each
direction between each pair of vertices.  A \emph{strongly connected component}
of a directed graph $G$ is a maximal strongly connected subgraph. The
collection of strongly connected components forms a partition of the set of
vertices of $G$, while it also has a \emph{topological ordering}, i.e., a
linear ordering of its components such that for every arc $(u,v)$, $u$ comes
before $v$ in the ordering. If each strongly connected component of $G$ is
contracted to a single vertex, the resulting graph is a directed acyclic graph
(DAG). The topic of this paper is \textsc{Directed $(p, q)$-Edge Dominating
Set}  ($(p , q)$-d\eds): given a directed graph $G=(V, E)$, a positive integer
$k$ and two non-negative integers $p, q$, we are asked to determine whether an
arc subset $K\subseteq E$ of size at most $k$ exists, such that every arc is
$(p, q)$-dominated by $K$.  Such a $K$ is called a {\em $(p, q)$-edge
dominating set} of $G$.

The \textsc{Dominating Set} problem is defined as follows: given an undirected graph $G=(V,E)$, we are asked to find a subset of vertices $D\subseteq V$, such that every vertex not in $D$ has at least one neighbor in $D$: $\forall v\notin D:N(v)\cap D\not=\emptyset$. For a directed graph $G=(V,E)$, every vertex not in $D$ is required to have at least one \emph{incoming} arc from at least one vertex of $D$: $\forall v\notin D:N^-(v)\cap D\not=\emptyset$.

We also use the \textsc{$k$-Multicolored
Clique} problem, which is defined as follows: given a graph $G=(V,E)$, with $V$
partitioned into $k$ independent sets $V=V_1\dot{\cup} \dots \dot{\cup} V_k$, where
$\forall i\in[1,k]$ it is  $|V_i|=n$,\footnote{We implicitly assume each set in the partition contains $n$ elements (rather than potentially fewer), as the numbering of vertices in each set will be used to encode algorithmic choices in our hardness proofs, whose descriptions will thus be more succint.} we are asked to find a subset $S\subseteq V$, such that
$G[S]$ forms a clique with $|S\cap V_i|=1,\forall i\in[1,k]$. The problem
\textsc{$k$-Multicolored Clique} is well-known to be
W[1]-complete (see \cite{FellowsHRV09}).

\paragraph{Complexity background:} We assume that the reader is familiar
with the basic definitions of parameterized complexity, such as the classes FPT
and W[1], as well as the Exponential Time Hypothesis (ETH), see~\cite{CyganFKLMPPS15}. 
For a problem $P$, we let $OPT_P$ denote the value of its optimal solution. We also make use of standard graph width measures,
such as the \emph{vertex cover number} $\vc$, \emph{treewidth} $\tw$ and \emph{pathwidth} $\pw$, whose definitions can also be found in \cite{CyganFKLMPPS15}. Formal definitions of notions related to approximation can be found in \cite{V01,Williamson2011} (also in appendices therein).

\paragraph{Treewidth and pathwidth:} A \emph{tree decomposition} of a graph $G=(V,E)$ is a pair $(\mathcal{X},T)$ with $T=(I,F)$ a tree and $\mathcal{X}=\{X_i|i\in I\}$ a family of subsets of $V$ (called \emph{bags}), one for each node of $T$, with the following properties:
 \begin{enumerate}[1)]
  \item $\bigcup_{i\in I}X_i=V$;
  \item for all edges $(v,w)\in E$, there exists an $i\in I$ with $v,w\in X_i$;
  \item for all $i,j,k\in I$, if $j$ is on the path from $i$ to $k$ in $T$, then $X_i\cap X_k\subseteq X_j$.
 \end{enumerate}
 The \emph{width} of a tree decomposition $((I,F),\{X_i|i\in I\})$ is
$\max_{i\in I}|X_i|-1$. The \emph{treewidth} of a graph $G$ is the minimum
width over all tree decompositions of $G$, denoted by $\textrm{tw}(G)$.  The
tree decomposition and width of a directed graph $G=(V,E)$ is defined as those
of the underlying graph of $G$, namely the undirected graph obtained from $G$
by forgetting the direction of arcs of $G$.

Moreover, for rooted $T$, let $G_i=(V_i,E_i)$ denote the \emph{terminal subgraph} defined by node $i\in I$, i.e.\ the induced subgraph of $G$ on all vertices in bag $i$ and its descendants in $T$. Also let $N_{i}(v)$ denote the neighborhood of vertex $v$ in $G_i$ and $dist_i(u,v)$ denote the distance between vertices $u$ and $v$ in $G_i$, while $dist(u,v)$ (absence of subscript) is the distance in $G$.

In addition, a tree decomposition can be converted to a \emph{nice} tree decomposition of the same width (in $O(\textrm{tw}^2\cdot n)$ time and with $O(\textrm{tw}\cdot n)$ nodes). The tree here is rooted and binary, while each node is one of the four types: 
\begin{enumerate}[a)]
 \item Leaf nodes $i$ are leaves of $T$ and have $|X_i|=1$;
 \item Introduce nodes $i$ have one child $j$ with $X_i=X_j\cup\{v\}$ for some vertex $v\in V$ and are said to \emph{introduce} $v$;
 \item Forget nodes $i$ have one child $j$ with $X_i=X_j\setminus\{v\}$ for some vertex $v\in V$ and are said to \emph{forget} $v$;
 \item Join nodes $i$ have two children denoted by $i-1$ and $i-2$, with $X_i=X_{i-1}=X_{i-2}$.
\end{enumerate}
Nice tree decompositions were introduced by \cite{Kloks94} and using them does not in general give any additional algorithmic possibilities, yet algorithm design becomes considerably easier.

Replacing ``tree'' by ``path'' in the above, we get the definition of \emph{pathwidth} $\pw$. We recall the following well-known relation:

\begin{lemma}\label{lem:relations} For any
graph $G$ we have $\tw(G)\le \pw(G)$.  
\end{lemma}

\paragraph{Tournaments:} A \emph{tournament} is a directed graph in which every pair of distinct vertices is connected by a single arc. Given a tournament $T$, 
we denote by $T^{rev}$ the tournament obtained from $T$ by reversing the direction of every arc. Every tournament has a \emph{king} (sometimes also called a $2$-king), being a vertex from which every other vertex can be reached by a path of length at most 2. One such king is the vertex of maximum out-degree (see \cite{BiswasJRS22}).  It is folklore that any tournament contains a \emph{Hamiltonian path}, being a directed path that uses every vertex. The \textsc{Dominating Set} problem can be solved by brute force in time $n^{O(\log n)}$ on tournaments, by the following lemma:

\begin{lemma}[\cite{CyganFKLMPPS15}]\label{lem_ds_tour}
 Every tournament on $n$ vertices has a dominating set of size $\le\log n+1$.
\end{lemma}

%
%

\section{Tractability}\label{sec_branch}

\subsection{FPT algorithms}\label{subsec:fpt}
In this section we present FPT branching algorithms for $(0,1)$-d\eds\ and $(1,1)$-d\eds. Both algorithms operate along similar lines, considering the particular ways available for domination of each arc.

\begin{theorem}\label{thm_9branching}
 The $(1,1)$-d\eds\ problem parameterized by solution size $k$ can be solved in time $O^*(9^k)$.
\end{theorem}
 \begin{proof}
We present an algorithm that works in two phases. In the first phase we perform
a branching procedure which aims to locate vertices with positive out-degree or
in-degree in the solution. The general approach of this procedure is standard
(as long as there is an uncovered arc, we consider all ways in which it may be
covered), and uses the fact that at most $2k$ vertices have positive in- or
out-degree in the solution. In order to speed up the algorithm, however, we use
a more sophisticated branching procedure that picks an endpoint of the current
arc $(u,v)$ and \emph{completely guesses} its behavior in the solution. This
ensures that this vertex will never be branched on again in the future. Once
all arcs of the graph are covered, we perform a second phase, which runs in
polynomial time, and by using a maximum matching algorithm finds the best
solution corresponding to the current branch.

 
Let us now describe the branching phase of our algorithm. We construct three
sets of vertices $V^+,V^-,V^{+-}$. The meaning of these sets is that when we
place a vertex $u$ in $V^+, V^-,$ or $V^{+-}$ we guess that $u$ has (i)
positive out-degree and zero in-degree in the optimal solution; (ii) positive
in-degree and zero out-degree in the optimal solution; (iii) positive in-degree
and positive out-degree in the optimal solution, respectively. Initially all
three sets are empty. When the algorithm places a vertex in one of these sets
we say that the vertex has been \emph{marked}.

Our algorithm now proceeds as follows: given a graph $G=(V,E)$ and three
disjoint sets $V^+, V^-, V^{+-}$, we do the following:

\begin{enumerate}

\item\label{it:br1} If $|V^+|+|V^-|+2|V^{+-}| > 2k$, reject. 

\item\label{it:br2} While there exists an arc $(u,v)$ with both endpoints
unmarked, do the following and return the best solution:
	
	\begin{enumerate}

	\item\label{it:br2a} Call the algorithm with $V^+ := V^+\cup\{v\}$ and the
other sets unchanged.  

	\item\label{it:br2b} Call the algorithm with $V^{+-} :=
V^{+-}\cup\{v\}$ and the other sets unchanged.

	\item\label{it:br2c} Call the algorithm with $V^- := V^-\cup\{u\}$ and the
other sets unchanged.  

	\item\label{it:br2d} Call the algorithm with $V^{+-} :=
V^{+-}\cup\{u\}$ and the other sets unchanged.

	\item\label{it:br2e} Call the algorithm with $V^+:= V^+\cup\{u\}$,
$V^-:= V^-\cup\{v\}$, and $V^{+-}$ unchanged.

	\end{enumerate}

\end{enumerate}

It is not hard to see that Step \ref{it:br1} is correct as
$|V^+|+|V^-|+2|V^{+-}|$ is a lower bound on the sum of the degrees of all
vertices in the optimal solution and therefore cannot surpass $2k$.

Branching Step \ref{it:br2} is also correct: in order to cover $(u,v)$, the
optimal solution must either take an arc coming out of $v$
(\ref{it:br2a},\ref{it:br2b}), or an arc coming into $u$
(\ref{it:br2c},\ref{it:br2d}), or, if none of the previous cases apply, it must
take the arc itself (\ref{it:br2e}).

Once we have applied the above procedure exhaustively, all arcs of the graph
have at least one marked endpoint. We say that an arc $(u,v)$ with $u\in
V^-\cup V^{+-}$, or with $v\in V^+\cup V^{+-}$ is covered. We now check if the
graph contains an uncovered arc $(u,v)$ with exactly one marked endpoint. We
then branch by considering all possibilities for its other endpoint. More
precisely, if $u\in V^+$ and $v$ is unmarked, we branch into three cases, where
$v$ is placed in $V^+$, or $V^-$, or $V^{+-}$ (and similarly if $v$ is the
marked endpoint). This branching step is also correct, since the degree
specification for the currently marked endpoint does not dominate the arc
$(u,v)$, hence any feasible solution must take an arc incident on the other
endpoint.

Once the above procedure is also applied exhaustively we have a graph where all
arcs either have both endpoints marked, or have one endpoint marked but in a
way that if we respect the degree specifications the arc is guaranteed to be
covered. What remains is to find the best solution that agrees with the
specifications of the sets $V^+, V^-, V^{+-}$.  

We first add to our solution $S$ all arcs $\delta(V^+,V^-)$, i.e., all arcs
$(u,v)$ such that $u\in V^+$ and $v\in V^-$, since there is no other way to
dominate these arcs. We then define a bipartite graph $H=(V^+\cup
V^{+-},V^-\cup V^{+-},\delta(V^+\cup V^{+-},V^-\cup V^{+-}))$.  That is, $H$
contains all vertices in $V^+$ along with a copy of $V^{+-}$ on one side, all
vertices of $V^-$ and a copy of $V^{+-}$ on the other side and all arcs in $E$
with tails in $V^+\cup V^{+-}$ and heads in $V^-\cup V^{+-}$. We now compute a
minimum edge cover of this graph, that is, a minimum set of edges that touches
every vertex. This can be done in polynomial time by finding a maximum matching
and then adding an arbitrary incident edge for each unmatched vertex. It is not
hard to see that a minimum edge cover of this graph corresponds exactly to the
smallest $(1,1)$-edge dominating set that satisfies the specifications of the
sets $V^+, V^-, V^{+-}$.

To see that the running time of our algorithm is $O^*(9^k)$, observe that
there are two branching steps: either we have an arc $(u,v)$ with both
endpoints unmarked; or we have an arc with exactly one unmarked endpoint. In
both cases we measure the decrease of the quantity $\ell:= 2k -
(|V^+|+|V^-|+2|V^{+-}|)$. The first case produces two instances with $\ell' :=
\ell-1$ (\ref{it:br2a},\ref{it:br2c}), and three instances with $\ell' := \ell
-2$. We therefore have a recurrence satisfying $T(\ell)\le 2T(\ell-1)+3T(\ell-2)$, which
gives $T(\ell)\le 3^\ell$. For the second case, we have three branches, all of
which decrease $\ell$ and we therefore also have $T(\ell)\le 3^\ell$ in this case.
Taking into account that, initially $\ell = 2k$ we get a running time of at
most $O^*(9^k)$. \end{proof}

\begin{theorem}\label{thm_2branching}
 The $(0,1)$-d\eds\ problem parameterized by solution size $k$ can be solved in time $O^*(2^k)$.
\end{theorem}
\begin{proof}
We give a branching algorithm that marks vertices of $V$. During the branching
process we construct three disjoint sets: $V_0$ contains vertices that will
have in-degree $0$ in the optimal solution; $V^+_F$ contains vertices that have
positive in-degree in the optimal solution and for which the algorithm has
already identified at least one selected incoming arc; and $V^+_?$ contains
vertices that have positive in-degree in the optimal solution for which we have not yet
identified an incoming arc. The algorithm will additionally mark some arcs as
``forced'', meaning that these arcs have been identified as part of the
solution.

Initially, the algorithm sets $V_0=V^+_F=V^+_?=\emptyset$. These sets will remain disjoint during
the branching. We denote $V^+=V^+_F\cup V^+_?$
and $V_r=V\setminus(V_0\cup V^+)$.

Before performing any branching steps we exhaustively apply the following
rules:

\begin{enumerate}

\item If $|V^+|>k$, we reject. This is correct since no solution can have more
than $k$ vertices with positive in-degree.

\item\label{it:2} If there exists an arc $(u,v)$ with $u,v\in V_0$, we reject.
Such an arc cannot be covered without violating the constraint that the
in-degrees of $u,v$ remain $0$.

\item\label{it:3} If there exists a source  $v\in V_r$, we set
$V_0:=V_0\cup\{v\}$. This is correct since a source will obviously have
in-degree $0$ in the optimal solution.

\item\label{it:4} If there exists an arc $(u,v)$ with $u\in V_0$ and $v\not\in
V^+_F$, we set $V^+_F:=V^+_F\cup\{v\}$ and $V^+_?:=V^+_?\setminus\{v\}$. This is
correct since the only way to cover $(u,v)$ is to take it. We mark all arcs
with tail $u$ as forced. 

\item\label{it:5} If there exists an arc $(u,v)$ with $v\in V_0$ and $u\not\in
V^+$, we set $V^+_?:=V^+_?\cup\{u\}$. This is correct, since we cannot cover
$(u,v)$ by selecting it (this would give $v$ positive in-degree).

\item\label{it:6} If there exists an arc $(u,v)$ with $v\in V^+_F$ and $u\in
V_r$ which is not marked as forced, then we set $V^+_?:=V^+_?\cup\{u\}$. We explain the correctness of this
rule below.
\end{enumerate}

The above rules take polynomial time and can only increase $|V^+|$.  We observe
that $V_r$ contains no sources (Rule \ref{it:3}). To see that Rule \ref{it:6}
is correct, suppose that there is a solution in which the in-degree of $u$ is
$0$, therefore the arc $(u,v)$ is taken. However, since $v\in V^+_F$, we have
already marked another arc that will be taken, so the in-degree of $v$ will end
up being at least $2$. Since $u$ is not a source (Rule \ref{it:3}), we replace
$(u,v)$ with an arbitrary incoming arc to $u$.  This is still a valid solution.

The first branching step is the following: suppose that there exists an arc
$(u,v)$ with $u,v\in V_r$. In one branch we set $V^+_?:= V^+_?\cup \{u\}$, and
in the other branch we set $V_0:=V_0\cup\{u\}$ and $V^+_F := V^+_F\cup\{v\}$ and
mark $(u,v)$ as forced. This branching is correct as any feasible solution will
either take an arc incoming to $u$ to cover $(u,v)$, or if it does not, will take
$(u,v)$ itself. In both branches the size of $V^+$ increases by~1.

Suppose now that we have applied all the above rules exhaustively, and that we
cannot apply the above branching step. This means that $(V_0\cup V^+)$ is a
vertex cover (in the underlying undirected graph). If there is a vertex $u\in V^+_?$ that has at least two in-neighbors
$v_1,v_2\in V_r$ we branch as follows: we either set $V^+_?:=
V^+_?\cup\{v_1\}$; or we set $V_0:=V_0\cup\{v_1\}$, $V^+_F:=V^+_F\cup\{u\}$,
and $V^+_?:=V^+_?\setminus\{u\}$ and mark the arc $(v_1,u)$ as forced. This is
correct, since a solution will either take an incoming arc to $v_1$, or the arc
$(v_1,u)$. The first branch clearly increases $|V^+|$. The key observation is
that $|V^+|$ also increases in the second branch, as Rule \ref{it:6} will
immediately apply, and place $v_2$ in $V^+_?$.

Suppose now that none of the above applies.  Because of Rule \ref{it:6} there
are no arcs from $V_r$ to $V^+_F$. Because the second branching Rule does not
apply, and because of Rule \ref{it:4}, each vertex $v\in V^+_?$ only has
in-neighbors in $V^+$ and at most one in-neighbor in $V_r$. For each $v\in V^+_?$
that has an in-neighbor $u\in V_r$ we select $(u,v)$ in the solution; for every
other $v\in V^+_?$ we select an arbitrary incoming arc in the solution; for
each $u\in V^+_F$ we select the incoming arcs that the branching algorithm has
identified. We claim that this is a valid solution. Because of Rule \ref{it:4}
all arcs coming out of $V_0$ are covered, because of Rule \ref{it:2} no arcs
are induced by $V_0$, and because of Rule \ref{it:5} all arcs going into $V_0$
have a tail with positive in-degree in the solution. We have selected in the solution every arc
from $V_r$ to $V^+_?$, and there are no arcs induced by $V_r$, otherwise we would
have applied the first branching rule. All arcs from $V_r$ to $V^+_F$ are marked as forced and we have selected them in the solution. 
Finally, all arcs with tail in $V^+$ are
covered. 

Because of the correctness of the branching rules, if there is a solution, one
of the branching choices will produce it. All rules can be applied in
polynomial time, or produce two branches with larger values of $|V^+|$. Since
this value never goes above $k$, we obtain an $O^*(2^k)$ algorithm. 
\end{proof}

 \subsection{Approximation algorithms}\label{subsec:approx}

We present here constant-factor approximation algorithms for $(0,1)$-d\eds, and $(1,1)$-d\eds. Both algorithms appropriately utilize a \emph{maximal} matching.

\begin{theorem}\label{prop:01approx}
 There are polynomial-time $3$-approximation algorithms for $(0,1)$-d\eds.
\end{theorem}
\begin{proof}
Let $G=(V,E)$ be an input directed graph. 
We partition $V$ into $(S,R,T)$ so that $S$ and $T$ are the sets of sources and sinks, respectively, and $R=V\setminus (S\cup T)$.
A $(0,1)$-edge dominating set $K$ is constructed as follows. 
\begin{enumerate}
\item Add the arc set $\delta^+(S)$ to $K$. 
\item For each vertex of $v\in (R\cap N^-(T))\setminus N^+(S)$, choose precisely one arc from $\delta^-(v)$ and add it to $K$. In other words, as long as there exists a vertex $v$ for which we have not yet selected any of its incoming arcs and which has an outgoing arc to a sink, we select arbitrarily an arc coming into $v$.

\item Let $G'=(R,E')$ be the subdigraph of $G$ whose arc set consists of arcs
not $(0,1)$-dominated by $K$ thus far constructed. Let $M$ be a set of arcs in $G'$, each corresponding to an edge of a maximal matching in the underlying undirected graph of $G'$ (using either direction) and $V(M)$ be the set of vertices
touched by $M$.
Let $M^-$ be the tails of the arcs in $M$ and let $I^+$ be the set of unmatched vertices $v$ which are not sinks in $G'$, that is, $v\in R\setminus
V(M)$ and $\delta^+_{G'}(v)\neq \emptyset$.  To $K$, we add all arcs of $M$, an
arbitrary incoming arc of $v$ for every $v\in M^-$, and an arbitrary incoming
arc of $v$ for every $v\in I^+$. 
\end{enumerate}

The above construction can be carried out in polynomial time. Furthermore, in
all steps where we add an arbitrary arc to a vertex $u$, we have $u\not\in S$,
therefore such an arc exists. Let us first observe that the constructed
solution is feasible.  Let $K_1$, $K_2$ and $K_3$ be the set of arcs added to
$K$ at step 1, 2 and 3, respectively.  $K_1$ contains all arcs incident on $S$,
so all these arcs are covered. For each arc $(u,v)$ with $v\in T$ we have
selected an arc going into $u$ to be put into $K_2$, so $(u,v)$ is covered.  Finally, for
each arc $(u,v)$ with $u,v\in R$ we consider the following cases: If $u\in
V(M)$ and $u$ is the head of an arc of $M$, then $(u,v)$ is covered since we
selected all arcs of $M$; if $u\in V(M)$ and $u$ is a tail of an arc in $M$
then $K_3$ contains an arc going into $u$, so $(u,v)$ is covered;  if $u\not\in
V(M)$ then $u\in I^+$, so we have selected an arc going into $u$. In all cases
$(u,v)$ is covered.

Let us now argue about the approximation ratio. Fix an optimal solution
$OPT_{(0,1)dEDS}$. First, note for $K_1=\delta^+(S)$ we must have $K_1\subseteq
OPT_{(0,1)dEDS}$, because the only arc that can $(0,1)$-dominate an arc of
$\delta^+(S)$ is itself. Let $OPT_2=OPT_{(0,1)dEDS}\setminus K_1$. 

Consider the set $R'=(R\cap N^-(T))\setminus N^+(S)$. We claim that for each
$v\in R'$ the set $OPT_2$ contains either at least one arc of $\delta^-(v)$ or
all arcs with tail $v$ and head in $T$. Let $OPT_2'$ be a set of arcs
constructed by selecting for each $v\in R'$ a distinct element of $OPT_2\cap
\delta^-(v)$, or if no such element exists all the arcs $(v,t)\in OPT_2$ with
$t\in T$. We have $|OPT_2'|\ge |K_2|$ because all vertices of $R'$ have an
out-neighbor in $T$. Let $OPT_3=OPT_2\setminus OPT_2'$. 

We will now argue that $|OPT_3|\ge|I^+|$. We first observe that any (optimal)
solution must contain at least one arc of $\delta_G^-(v)\cup \delta_G^+(v)$ for
every $v\in I^+$.  In order to justify step 3, the following claim provides a
key observation.

\begin{claim}
It holds that $\delta(S,I^+)=\delta(I^+,T)=\emptyset$. Furthermore $I^+$ is an independent set in the underlying undirected graph of $G$. 
\end{claim}
\begin{proof}
If there is an arc from $s\in S$ to $v\in I^+$ then $(s,v)\in K_1$, which
implies that all arcs coming out of $v$ are dominated by $K_1$. This means that
$v$ is a sink in $G'$, which is a contradiction. If there is an arc from $v\in
I^+$ to $t\in T$ then there is an arc going into $v$ that belongs to $K_2$,
which again makes $v$ a sink in $G'$, contradiction. Therefore,
$\delta(S,I^+)=\delta(I^+,T)=\emptyset$.

Suppose that $I^+$ is not an independent set in $G$ and let $(u,v)$ be an arc
with $u,v\in I^+$. However, $M$ is maximal and $u,v$ are unmatched, which
implies that the arc $(u,v)$ does not appear in $G'$. This means that either
$(u,v)\in K_2$, which makes $v$ a sink in $G'$, or an arc going into $u$
belongs in $K_1\cup K_2$, which makes $u$ a sink in $G'$. In both cases we have
a contradiction.  \end{proof}

Let us now use the above claim to show that $|OPT_3|\ge |I^+|$. First, observe
that $I^+\cap R'=\emptyset$, as all vertices of $R'$ are sinks in $G'$.
Furthermore, all arcs of $OPT_2'$ have their heads in $R'\cup T$, hence none of
them have their heads in $I^+$. Similarly, no arc of $K_1$ has its head in
$I^+$, because this would make its head a sink in $G'$. Therefore, all arcs
with tail in $I^+$ that exist in $G'$ are dominated by $OPT_3$. We now observe
that since $I^+$ is an independent set, no arc of $OPT_3$ can dominate two arcs
with tails in $I^+$. Therefore, $|OPT_3|\ge |I^+|$.
 We now have 

\[|K_1|+|K_2|+|I^+|\leq |K_1|+|OPT_2'|+|OPT_3| \leq |OPT_{(0,1)dEDS}|.\] 

In order to $(0,1)$-dominate the entire arc set $M$, one needs to take at least
$|M|$ arcs, because $M$ corresponds to a matching in the underlying undirected graph and we thus have $|OPT_{(0,1)dEDS}|\ge |M|$.
We also recall the definition of $K_3$: it contains all arcs of $M$, one
arbitrary incoming arc of each $v\in M^-$, and an arbitrary incoming
arc of each $v\in I^+$.
We therefore deduce \[|K|\leq |K_1|+|K_2|+2|M|+|I^+| \leq
3|OPT_{(0,1)dEDS}|.\]\end{proof}

\begin{theorem}\label{thm_approx}
 There is a polynomial-time $8$-approximation algorithm for $(1,1)$-d\eds.
\end{theorem}
 \begin{proof}
 Let $G=(V,E)$ be an input directed graph. We partition $V$ into $(S,R,T)$ so
that $S$ and $T$ are the sets of sources and sinks, respectively, and
$R=V\setminus (S\cup T)$.

 We construct an $(1,1)$-edge dominating set $K$ as follows. 
 \begin{enumerate}
 \item Add the arc set $\delta(S,T)$ to $K$. 
 \item For each vertex of $v\in R\cap N^+(S)$, choose precisely one arc from $\delta^+(v)$ and add it to $K$.
 \item For each vertex of $v\in R\cap N^-(T)$, choose precisely one arc from $\delta^-(v)$ and add it to $K$.
 \item Let $G'=(R,E')$ be the subdigraph of $G$ whose arc set consists of those arcs not $(1,1)$-dominated by $K$ thus far constructed. Let $M$ be a set of arcs in $G'$, each corresponding to an edge (any direction) of a maximal matching in the underlying 
 graph of $G'$. Let $M^-$ and $M^+$ be the tails and heads of the arcs in $M$, respectively. 
 To $K$, we add all arcs of $M$, an arc of $\delta_G^-(v)$ for every $v\in M^-$, and also an arc of $\delta_G^+(v)$  for every $v\in M^+$. 
 \end{enumerate} 
 Clearly, the algorithm runs in polynomial time. In particular, for any vertex $v$ considered in Steps 2-4, both $\delta^+(v)$ and $\delta^-(v)$ are non-empty and choosing an arc from a designated set is always possible.  We show that $K$ is indeed an $(1,1)$-edge dominating set.
 Suppose that an arc $(u,v)$ is not $(1,1)$-dominated by $K$. As the first, second and third step of the construction ensures 
 that any arc incident with $S\cup T$ is $(1,1)$-dominated, we know that $(u,v)$ is contained in the subdigraph $G'$ constructed at step 4. 
 For $(u,v)\notin M$ and $M$ corresponding to a maximal matching, one of the vertices $u,v$ must be incident with $M$. 
 Without loss of generality, we assume $v$ is incident with $M$ (and the other cases are symmetric). If $v\in M^-$, then clearly the arc $e\in M$ 
 whose tail coincides with $v$ would $(1,0)$-dominate $(u,v)$, a contradiction. 
 If $v\in M^+$, then the outgoing arc of $v$ added to $K$ at step 4 would $(1,0)$-dominate $(u,v)$, again reaching a contradiction. 
 Therefore, the constructed set $K$ is a solution to $(1,1)$-d\eds.
 
 To prove the claimed approximation ratio, we first note that $\delta(S,T)$ is contained in any (optimal) solution because any arc of $\delta(S,T)$ can be $(1,1)$-dominated only by itself. Note that 
 these arcs do not $(1,1)$-dominate any other arcs of $G$.
 Further, we have  $|R\cap N^+(S)|\leq OPT_{(1,1)dEDS}-|\delta(S,T)|$ because in order to $(1,1)$-dominate any arc of the form $(s,r)$ with $s\in S$ and $r\in R$, 
 one must take at least one arc from $\{(s,r)\}\cup \delta^+(r)$. Since the sets $\{(s,r):s\in S\}\cup \delta^+(r)$ are disjoint over all $r\in R\cap N^+(S)$, 
 the inequality holds.
 Likewise, it holds that $|R\cap N^-(T)|\leq OPT_{(1,1)dEDS}-|\delta(S,T)|$.
 In order to $(1,1)$-dominate the entire arc set $M$, one needs to take at least $|M|/2$ arcs. This is because an arc $e$ can $(1,1)$-dominate at most two arcs of $M$. 
 That is, we have $|M|/2 \leq  OPT_{(1,1)dEDS}-|\delta(S,T)|$.
 Therefore, it is $|K|\leq |\delta(S,T)|+|R\cap N^+(S)|+|R\cap N^-(T)|+3|M|\leq 8\, OPT_{(1,1)dEDS}$.
 \end{proof}

\subsection{Polynomial kernels}\label{sec_kernel}

We give polynomial kernels for $(1,1)$-d\eds\ and $(0,1)$-d\eds.
We first introduce a relation between the vertex cover number and the size of a minimum $(1,1)$-edge dominating set, shown by \cite{HanakaNO17} (as a corollary to their Lemma~22) and then proceed to show a quadratic-vertex/cubic-edge kernel for $(1,1)$-d\textsc{EDS}.

\begin{lemma}[\cite{HanakaNO17}]\label{vc:11EDS}
Given a directed graph $G$, let $G^*$ be the undirected underlying graph of $G$,  $\vc(G^*)$ be the vertex cover number of $G^*$, and $K$ be a minimum $(1,1)$-edge dominating set in $G$.
Then  $\vc(G^*)\le 2|K|$.
\end{lemma}

\begin{theorem}\label{prop:kernel11}
There exists an $O(k^2)$-vertex/$O(k^3)$-edge kernel for $(1,1)$-d\eds.
\end{theorem}
 \begin{proof}
 Given a directed graph $G$, we denote the underlying undirected graph of $G$ by $G^*$. Let $K$ be a minimum $(1, 1)$-edge dominating set and $\vc(G^*)$ be the size of a minimum vertex cover in $G^*$. 
 First, we find a \emph{maximal} matching $M$ in $G^*$.
 If $|M|>2k$, we conclude this is a no-instance by Lemma~\ref{vc:11EDS} and the well-known fact that $|M|\le \vc(G^*)$, see~\cite{GJ76}. 
 Otherwise, let $S$ be the set of endpoints of edges in $M$. 
 Then $S$ is a vertex cover of size at most $4k$ for the underlying undirected graph of $G$ and $V\setminus S$ is an independent set.
  
 We next explain the reduction step.
 For each $v\in S$, we arbitrarily mark the first $k+1$ tail vertices of incoming arcs of $v$  with ``in'' (or all, if the in-degree of $v$ is $\le k$) and also arbitrarily the first $k+1$ head vertices of outgoing arcs of $v$ with ``out'' (or all, if the out-degree of $v$ is $\le k$).
 After this marking, if there exists a vertex $u\in V\setminus S$ without marks ``in'', ``out'', we can delete it.
 
We next show the correctness of the above.  First, we can observe that  if some
$v\in S$ has more than $k+1$ incoming arcs, then any feasible solution of size
at most $k$ must select an arc with tail $v$.  Similarly, if $v\in S$ has more
than $k+1$ outgoing arcs, any feasible solution of size at most $k$ must select
an arc with head $v$. Consider now an unmarked vertex $u$ and suppose that it
is the tail of an arc $(u,v)$ with $v\in S$ (the case where $u$ is the head is
symmetric). The vertex $v$ has $k+1$ other incoming arcs, besides $(u,v)$,
otherwise $u$ would have been marked. Therefore, in any solution of size at
most $k$ in the graph where $u$ has been deleted we must select an arc coming
out of $v$. This arc dominates $(u,v)$. Therefore, any feasible solution of the
new graph remains feasible in the original graph. For the other direction,
suppose a solution for the graph $G$ selects the arc $(u,v)$. We consider the
same solution without $(u,v)$ in the graph where $u$ is deleted. If this is
already feasible, we are done. If not, any non-dominated arc must have $v$ as
its tail (every other arc dominated by $(u,v)$ has been deleted). All these
arcs can be dominated by adding to the solution an arc going into $v$. Note also that any deleted vertex $u\in V\setminus S$ is only connected to vertices in $S$, since $S$ is a vertex cover and the above thus accounts for all possibly deleted arcs.
 
After exhaustively applying the above rule every vertex of the independent set
will be marked. We mark at most $2(k+1)$ vertices of the independent set for
each of the at most $4k$ vertices of $S$, so we have a total of at most
$8k^2+12k$ vertices.  Moreover, there exist at most $4k\cdot
(8k^2+8k)=32k^3+32k^2$ arcs between the sets of the vertex cover and the
independent set.  Therefore, the number of arcs in the reduced graph is at most
${{4k}\choose{2}}+32k^3+32k^2=32k^3 + 32k^2 + 2k(4k-1)= O(k^3)$.  \end{proof}

 Next, we  note that
 the size of a minimum $(0,1)$-edge dominating set is equal to, or greater than the size of a minimum $(1,1)$-edge dominating set.
 Thus, we have $|M|\le \vc(G^*)\le  2|K|$ where $K$ is a $(0,1)$-edge dominating set and $M$ is a maximal matching.
 We give a more strict relation between $\vc$ and the size of a minimum $(0,1)$-edge dominating set, however, that is then used to obtain Theorem~\ref{prop:kernel01}.

\begin{lemma}\label{vc:01EDS}
Given a directed graph $G$, let $G^*$ be the undirected underlying graph of $G$,  $\vc(G^*)$ be the vertex cover number of $G^*$, and $K$ be a minimum $(0,1)$-edge dominating set in $G$.
Then  $\vc(G^*)\le |K|$.
\end{lemma}
\begin{proof}
For an arc $(u,v)$, the head vertex $v$ covers all arcs (i.e., edges) dominated by $(u,v)$ in $G^*$. 
Since $K$ dominates all edges in $G$, the set of head vertices of $K$ is a vertex cover in $G^*$.
Thus, $\vc(G^*)\le |K|$.
\end{proof}


\begin{theorem}\label{prop:kernel01}
There exists an $O(k)$-vertex/$O(k^2)$-edge kernel for $(0,1)$-d\eds.
\end{theorem}
\begin{proof}
Our first reduction rule states that if there exists an arc $(s,t)$ where $s$
is a source ($d^-(s)=0$) and $t$ is a sink ($d^+(t)=0$) then we delete this arc
and set $k:=k-1$. This rule is correct because the only arc that dominates
$(s,t)$ is the arc itself, and $(s,t)$ does not dominate any other arc. In the
remainder we assume that this rule has been applied exhaustively.

We then find a \emph{maximal} matching $M$ in the underlying undirected graph. If
$|M|>k$, then by Lemma~\ref{vc:01EDS} we conclude that we can reject.
Otherwise, the set of vertices incident on $M$, denoted by $S$ is a vertex
cover of size at most $2k$ and $V\setminus S$ is an independent set.

Now, suppose that there exist $k+1$ vertices in $V\setminus S$ with positive
out-degree. This means that there exist $k+1$ arcs with distinct tails in
$V\setminus S$, and heads in $S$. No arc of the graph dominates two of
these arcs (since $V\setminus S$ is independent), therefore any feasible
solution has size at least $k+1$ and we can reject. 

We can therefore assume that the number of non-sinks in $V\setminus S$ is at
most $k$.  We will now bound the number of sinks.  Let $T$ be the set of sinks,
that is, $T$ contains all vertices $v$ for which $d^+(v)=0$. We edit the graph
as follows: delete all vertices of $T\setminus S$; add a new vertex $u$
which is initially not connected to any vertex; and then for each vertex $v\in
S$ such that there is an arc $(v,t)$ with $t\in T\setminus S$ in $G$ we
add the arc $(v,u)$.  We claim that this is an equivalent instance.

Before arguing for correctness, we observe that the new instance has at most $3k+1$
vertices: $S$ has at most $2k$ vertices, $V\setminus S$ has at most $k$
non-sinks, and all sinks of $V\setminus S$ have been replaced by $u$.  This
graph clearly has $O(k^2)$ edges.

Let $G$ be the original graph and $G'$ the graph obtained after replacing all
sinks in the independent set with the new vertex $u$.  Consider an optimal
solution in $G$.  If the solution contains an edge $(v,t)$ where $t\in
V\setminus S$ is a sink, then we know that $v$ is not a source (otherwise we
would have simplified the instance by deleting $(v,t)$). We edit the solution
by replacing $(v,t)$ with an arbitrary arc incoming to $v$. Repeating this
gives a solution which does not include any arc whose head is a sink of
$V\setminus S$, but for each such arc $(v,t)$ contains an arc going into
$v$.  This is therefore a valid solution of $G'$, as it dominates all arcs
going into $u$.  For the converse direction we similarly edit a solution to
$G'$ by replacing any arc $(v,u)$ with an arbitrary arc going into $v$ (again,
we can safely assume that such an arc exists). The result is a valid solution
for $G$ with the same size.  \end{proof}

\section{Treewidth}\label{sec_whard_tw}

In this section we characterize the complexity of $(p,q)$-d\eds\ parameterized
by the treewidth of the underlying graph of the input. Our main result is that, even though the problem is FPT when
parameterized by $p+q+\tw$, it becomes W[1]-hard if parameterized only by $\tw$
(in fact, also by $\pw$), even if we add the size of the optimal solution as a
second parameter.  The algorithm is based on standard dynamic programming
techniques, while for hardness we reduce from the well-known W[1]-complete
\textsc{$k$-Multicolored Clique} problem (\cite{FellowsHRV09}).

\subsection{Hardness for Treewidth}
 
\paragraph{Construction:} Before we proceed, let us define a more general
version of $(p,q)$-d\eds\ which will be useful in our reduction. Suppose that
in addition to a digraph $G=(V,E)$ we are also given as input a subset
$I\subseteq E$ of ``optional'' arcs. In \textsc{Optional} $(p,q)$-d\eds\ we are
asked to select a minimum set of arcs that dominate all arcs of $E\setminus I$,
meaning it is not mandatory to dominate the optional arcs. We will describe a
reduction from $k$-\textsc{Multicolored Clique} to a special instance of
\textsc{Optional} $(p,q)$-d\eds, and then show how to reduce this to the
original problem without significantly modifying the treewidth or the size of
the optimum.

Given an instance $[G=(V,E),k]$ of \textsc{$k$-Multicolored Clique}, with
$V=\bigcup_{i\in[1,k]}V_i$ and $V_i=\{v^i_0,\dots,v^i_{n-1}\}$, where we assume that $n$ is even (without loss of generality) we will construct an instance
$G'=(V',E')$ of \textsc{Optional} $(p,q)$-d\eds.  We set $p=q=3n$. We
begin by adding to $V'$ all vertices of $V$ and connecting each set $V_i$ into
a directed cycle of length $n$.  Concretely, we add the arcs $(v^i_j,
v^i_{j+1})$ for all $i\in[1,k]$ and $j\in[0,n-1]$ where addition is performed
modulo $n$.

Intuitively, the idea up to this point is that selecting the vertex $v^i_j$ in
the clique is represented in the new instance by selecting the arc of the cycle
induced by $V_i$ whose head is $v^i_j$.  In order to make it easier to prove
that the optimal solution will be forced to select one arc from each directed
cycle we add to our instance the following: for each $i\in[1,k]$ we construct a
directed cycle of length $5n+1$ and identify one of its vertices with
$v^i_{n/2}$.  We call these $k$ cycles the ``guard'' cycles.

Finally, we need to add some gadgets to ensure that the arcs selected really
represent a clique. For each pair of vertices of $G$, $v^i_a, v^j_b$ which are
not connected by an edge in $G$ we do the following (depending on the values of $a,b$): we first construct two new
vertices $e_{i,j,a,b}, f_{i,j,a,b}$ and an arc $(e_{i,j,a,b},f_{i,j,a,b})$
connecting them. Then for the ``forward'' paths, if $a>0$ we construct a directed path of length $a+2n$ from
$v^i_0$ to $e_{i,j,a,b}$; if $b>0$ we construct a directed path of length
$b+2n$ from $v^j_0$ to $e_{i,j,a,b}$.
For the ``backward'' paths, if $a>0$ we construct a path of length
$3n-a+1$ from $f_{i,j,a,b}$ to $v^i_0$, otherwise we make a path of length
$2n+1$ from $f_{i,j,a,b}$ to $v^i_0$; if $b>0$ we construct a path of length
$3n-b+1$ from $f_{i,j,a,b}$ to $v^j_0$, otherwise we make a path of length
$2n+1$ from $f_{i,j,a,b}$ to $v^j_0$.

To complete the instance we define all arcs of the cycles induced by the sets
$V_i$, all arcs of the guard cycles, and all arcs of the form $(e_{i,j,a,b},
f_{i,j,a,b})$ as mandatory, and all other arcs (that is, internal arcs of the
paths constructed in the last part of our reduction) as optional. See Figure~\ref{fig:DeDS-Wtw-construction} for an example.

\begin{figure}[htbp]
\centerline{\includegraphics[width=120mm]{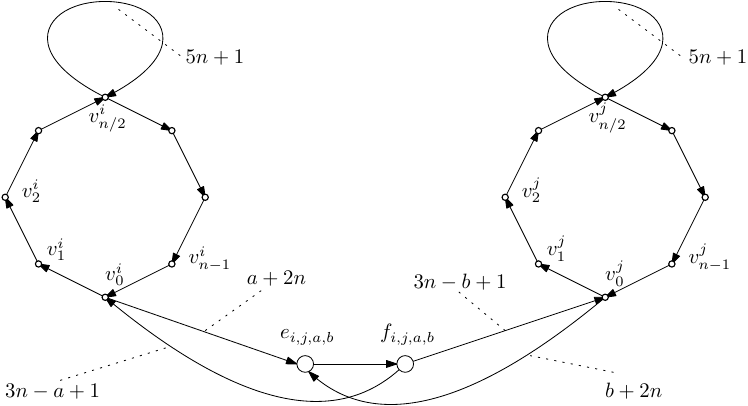}}
\caption{An example of our construction, where dotted lines show the length of each path.}
\label{fig:DeDS-Wtw-construction}
\end{figure}

\begin{lemma}\label{lem:twhard1} If $G$ has a multi-colored clique of size $k$,
then $G'$ has a partial $(3n,3n)$-d\eds\ of size~$k$. \end{lemma}

\begin{proof}
Suppose there is a multi-colored clique in $G$ of size $k$ that selects the vertex $v^i_{f(i)}$ for each
$i\in[1,k]$. We select in $G'$ the $k$ arcs
$(v^i_{f(i)-1}, v^i_{f(i)})$, where $f(i)-1$ is computed modulo $n$, that is,
the $k$ arcs of the cycles induced by $\bigcup_{i\in[1,k]}V_i $ whose heads
coincide with the vertices of the clique.

Let us see why this set of $k$ arcs $(3n,3n)$-dominates all non-optional arcs.
It is not hard to see that these arcs dominate the $k$ cycles induced by
$\bigcup_{i\in[1,k]}V_i$. For the guard cycles, for any $j\in[0,n-1]$ consider the
arc $(v^i_{j-1},v^i_j)$, where again $j-1$ is computed modulo $n$. We claim
that this arc dominates all the arcs of the guard cycle. To see this, suppose
first that $1\le j\le n/2$. Then, there are $5n/2+j$ arcs of the guard cycle
that lie on a path of length at most $3n$ from $v^i_j$ (because the distance
from the head of the selected arc to $v^i_{n/2}$ is $n/2-j$), and $5n/2-(j-1)$
arcs of the guard cycle that lie in a path of length at most $3n$ to
$v^i_{j-1}$ (because the distance from $v^i_{n/2}$ to the tail of the selected
arc is $n/2+j-1$).  These two sets are disjoint, so the total number of
dominated arcs in the cycle is $5n+1$.  The reasoning is similar if $j>n/2$ or
$j=0$.

Finally, let us see why this set dominates all arcs of the form
$(e_{i,j,a,b},f_{i,j,a,b})$, where $v^i_a, v^j_b$ are not connected in $G$.
Since these two vertices are not connected, we have either $f(i)\neq a$ or
$f(j)\neq b$.  Suppose without loss of generality that $f(i)= a'\neq a$ (the
other case is symmetric). We now consider the following cases:

\begin{enumerate}

\item If $a=0$, then since $a'\neq a$ we have $0< a'\le n-1$. Recall that if
$a=0$ we have a path of length $2n+1$ from $f_{i,j,a,b}$ to $v^i_0$, while the
path from $v^i_0$ to $v^i_{a'-1}$ has length at most $n-2$. Therefore, the
length of the path from $f_{i,j,a,b}$ to the tail $v^i_{a'-1}$ of the selected
arc is at most $3n-1$ and the arc $(e_{i,j,a,b},f_{i,j,a,b})$ is dominated.

\item If $a'=0$, then since $a'\neq a$ we have $0<a\le n-1$. In this case there
is a path of length $a+2n\le 3n-1$ from $v^i_0$ to $e_{i,j,a,b}$. Since $v^i_0$ is
the head of a selected arc, the arc $(e_{i,j,a,b},f_{i,j,a,b})$ is dominated.

\item  If $n-1\ge a'>a>0$, then we observe that there is a path from $v^i_a$ to
$e_{i,j,a,b}$ of length exactly $3n$: the distance from $v^i_a$ to $v^i_0$ is
$n-a$ and we have added a path of length $a+2n$ from $v^i_0$ to $e_{i,j,a,b}$. If
$a'>a$ then the path from $v^i_{a'}$ to $e_{i,j,a,b}$ is shorter than $3n$, so the
arc $(e_{i,j,a,b}, f_{i,j,a,b})$ is dominated. 

\item Finally, if $n-1\ge a>a'>0$, then we recall that there is a path from
$f_{i,j,a,b}$ to $v^i_0$ of length $3n-a+1$, and there is a path from $v^i_0$
to $v^i_{a'-1}$ of length $a'-1$, so the path from $f_{i,j,a,b}$ to the tail of
the selected arc is at most $3n-a+1+a'-1<3n$ and the arc
$(e_{i,j,a,b},f_{i,j,a,b})$ is dominated.
\end{enumerate}
\end{proof}

\begin{lemma}\label{lem:twhard2} If $G'$ has a partial $(3n,3n)$-d\eds\ of size
$k$, then $G$ has a multi-colored clique of size~$k$. \end{lemma}

\begin{proof}
We first argue that any valid solution must contain for each $i\in[1,k]$ an arc
of the form $(v^i_j,v^i_{j+1})$, where addition is done modulo $n$, or some arc
from the guard cycle. Suppose that this is not the case for some $i$. We then
argue that there is an arc of the guard cycle that is not dominated. In
particular, consider the arc $(u,v)$ of the guard cycle such that $u$ is at
distance exactly $5n/2$ from $v^i_{n/2}$. Observe that the path from $v$ to
$v^i_{n/2}$ also has length $5n/2$.
We argue that this arc is not dominated.
Indeed, for any selected arc $(u',v')$, the path from $v'$ to $u$ goes through
$v^i_0$, since we have not selected any arcs from inside the two cycles.
The distance from $v^i_0$ to $u$ is already exactly $3n$, however, so $(u',v')$
does not $(0,3n)$ dominate $(u,v)$. Similarly, $(u',v')$ does not $(3n,0)$
dominate $(u,v)$ because the distance from $v$ to $v^i_0$ (which lies on a
shortest path from $v$ to $u'$) is $3n$.

Because of the above, we know that a solution that selects exactly $k$ arcs
must select exactly one arc from each cycle induced by a $V_i$ or its attached
guard cycle. Let us also argue that we may assume that the solution does not
select any arcs from the guard cycles. Suppose for contradiction that a
solution selects $(u,v)$ from a guard cycle. We have either
$dist(v^i_{n/2},u)\ge 5n/2$ or $dist(v,v^i_{n/2})\ge 5n/2$. In the former case
the arc $(u,v)$ does not $(3n,0)$ dominate any arc with endpoints outside $V_i$
and its guard cycle, because to do so, the dominated arc would have to lie in a
path of length at most $3n$ going into $u$. Such a path must go through
$v^i_0$, and the distance from $v^i_0$ to $u$ is already at least $3n$. Since
$(u,v)$ may only $(0,3n)$ dominate arcs outside $V_i$, we replace $(u,v)$ with
$(v^i_{n-1}, v^i_0)$, which dominates all arcs inside the two cycles and
$(0,3n)$ dominates more arcs than $(u,v)$ outside the cycles. Similarly, in the
other case we replace the selected arc with $(v^i_0,v^i_1)$, which $(3n,0)$
dominates more arcs outside the cycles.

We therefore assume that the solution selects exactly one arc from each cycle
induced by a $V_i$. Let $f(i)$, for $i\in[1,k]$ be the head of the selected arc
in the cycle induced by $V_i$. We claim that the set $\{ v^i_{f(i)}\ |\
i\in[1,k]\ \}$ is a multi-colored clique in $G$.

Suppose that $f(i)=a, f(j)=b$ and $v^i_a, v^j_b$ are not connected. We argue
that the arc $(e_{i,j,a,b},f_{i,j,a,b})$ (which, by construction, exists in
$G'$) is not dominated by our supposed solution, which will give a
contradiction. Observe that the endpoints of the arc
$(e_{i,j,a,b},f_{i,j,a,b})$ are at distance at least $4n$ from each
$v^{\ell}_0$, for any $\ell\not\in\{i,j\}$. As a result, the only selected arcs
that could dominate $(e_{i,j,a,b},f_{i,j,a,b})$ are the selected arcs with
heads $v^i_a, v^j_b$. However, $(v^i_{a-1}, v^i_a)$ does not $(0,3n)$ dominate
the arc in question: the distance from $v^i_a$ to $e_{i,j,a,b}$ is $3n$
(distance $n-a$ from $v^i_a$ to $v^i_0$ and $2n+a$ from $v^i_0$ to
$e_{i,j,a,b}$); $(v^i_{a-1}, v^i_a)$ does not $(3n,0)$ dominate the arc in
question: the distance from $f_{i,j,a,b}$ to $v^i_{a-1}$ is $3n$ (if $a>0$ we
have distance $a-1$ from $v^i_0$ to $v^i_{a-1}$  and $3n-a+1$ from
$f_{i,j,a,b}$ to $v^i_0$, while if $a=0$ we have distance $n-1$ from $v^i_0$ to
the tail of the selected arc and distance $2n+1$ from $f_{i,j,a,b}$ to
$v^i_0$). By identical arguments, $(v^j_{b-1}, v^j_b)$ does not dominate the
arc $(e_{i,j,a,b},f_{i,j,a,b})$, so we have a contradiction.
\end{proof}

\begin{lemma}\label{lem:twhard3} The pathwidth of (the underlying graph of)
$G'$ is at most $2k+3$.  Furthermore, there exists a set of vertices $S$ of
$G'$ that contains no sources or sinks such that (i) all optional arcs are
incident to a vertex of $S$ and (ii) for each $u\in S$ all arcs incindent on
$u$ are optional.  \end{lemma}

\begin{proof}
For the pathwidth bound, it is a well-known fact that deleting a vertex from a
graph decreases the pathwidth by at most one (since this vertex may be added to
all bags in a decomposition of the new graph). Hence, we begin by deleting from
the graph the $2k$ vertices $\{ v^i_0, v^i_{n/2}\ |\ i\in[1,k]\ \}$. The graph
becomes a forest, and its pathwidth is upper-bounded by the maximum pathwidth
of any of its component trees. These trees are either paths or trees with two
vertices of degree higher than $2$ (these are the vertices $e_{i,j,a,b},
f_{i,j,a,b}$), but such trees are easily seen to have pathwidth at most $3$.

For the second claim we observe that the optional arcs are exactly the arcs
that were added in directed paths connecting $v^i_0$ to $e_{i,j,a,b},
f_{i,j,a,b}$, for some $i,j,a,b$. We therefore define $S$ to be the set of
internal vertices of such paths.
\end{proof}

\begin{theorem}\label{thm:tw_whard} $(p,q)$-d\eds\ is W[1]-hard parameterized
by the pathwidth $\pw$ of the underlying graph and the size $k$ of the optimum.
Furthermore, if there is an algorithm solving $(p,q)$-d\eds\ in time
$n^{o(\pw+k)}$, then the ETH is false. \end{theorem}

\begin{proof}
We start with an instance of \textsc{Multi-Colored Clique} and use Lemmas
\ref{lem:twhard1}, \ref{lem:twhard2}, \ref{lem:twhard3} to construct an
equivalent instance of \textsc{Optional} $(3n,3n)$-d\eds, with pathwidth $O(k)$
and optimal solution target $k$. What remains is to show how to transform this
into an equivalent instance of the standard version of d\eds, without affecting
the pathwidth or the size of the optimal solution too much. The theorem will then follow from standard facts about \textsc{Multi-Colored Clique}, namely that
the problem is W[1]-hard and not solvable in $n^{o(k)}$ under the ETH.

Given the instance $G'$ of \textsc{Optional} $(3n,3n)$-d\eds, we add to the
graph two new vertices $u_1,u_2$ and an arc $(u_1,u_2)$. We construct $k+2$
directed paths of $3n$ arcs (using new vertices). For each such path, we
identify its last vertex (sink) with $u_1$. Recall that there is a set of
vertices $S$ that is incident on all optional arcs. For each $u\in S$ we do the
following: we construct a new directed path of length $3n-1$ from $u_2$ to $u$;
and we construct a new directed path of length $3n-1$ from $u$ to $u_1$. We
claim that the new instance has a $(3n,3n)$-dominating set of size $k+1$ if and
only if the \textsc{Optional} d\eds\ instance has a solution of size $k$.

Suppose there is a solution of size $k$ that dominates all mandatory arcs of
$G'$. In the new instance we select the same arcs, as well as $(u_1,u_2)$. We
claim that $(u_1,u_2)$ dominates all the new arcs we added, since all such arcs
belong in a path of length at most $3n$ going into $u_1$ or coming out of
$u_2$. Furthermore, $(u_1,u_2)$ dominates all optional arcs of $G'$, since for
each such arc there exists $u\in S$ such that the arc is incident on $u$, and
$u$ is at distance at most $3n-1$ from $u_2$ and to $u_1$.

Suppose that there is a solution of size $k+1$ for the new instance. We first
claim that this solution must contain $(u_1,u_2)$. Indeed, consider the $k+2$
arcs incident on the sources of the paths whose sinks we identified with $u_1$.
No other arc of the instance dominates more than one of the arcs incident on
these sources. Hence, if we do not select $(u_1,u_2)$, we must have a solution
of size at least $k+2$. Now, assume that $(u_1,u_2)$ has been selected and note
that, as argued above, this arc dominates all new arcs as well as all optional
arcs. Furthermore, observe that $(u_1,u_2)$ does not dominate any non-optional
arc of $G'$, since its distance to any vertex of $V\setminus S$ is at least
$3n$ in both directions, and all arcs incident on $S$ are optional.

Suppose that the solution also contains another arc that does not appear in
$G'$. We claim that we can always replace this with another arc that appears in
$G'$. Indeed, an arc from the $k+2$ paths going into $u_1$ is redundant (the
arc $(u_1,u_2)$ dominates more arcs); an arc from a path from $u_2$ to $u\in S$
can be replaced by any arc of $G'$ going into $u$ (such an arc exists since $u$
is not a sink); and an arc from a path from $u\in S$ to $u_1$ may be replaced
by another arc coming out of $u$ in $G'$. The latter two replacements are
correct because the new arcs dominate more arcs of $G'$, while all arcs which
do not appear in $G'$ have already been dominated by $(u_1,u_2)$. We therefore
arrive at a set of at most $k$ arcs of $G'$. As argued above $(u_1,u_2)$ does
not dominate any of the mandatory arcs of $G'$. Furthermore, for any two
vertices $u,v$ of $G'$ such that $dist(u,v)\ge 3n$ in $G'$ we still have
$dist(u,v)\ge 3n$ in the new instance, as all paths we have added have length
at least $3n-1$.  This means that if the $k$ arcs of $G'$ we have selected in
the new instance dominate all mandatory arcs, they also dominate them in $G'$.

Finally, it is not hard to see that the pathwidth of the new graph remains
$O(k)$. We delete $u_1,u_2$ from the graph and now the resulting graph is $G'$
with the addition of some path components and also some pendant paths attached
to the vertices of $S$. We can construct a path decomposition of the new graph
by taking a path decomposition of $G'$ and, for each $u\in S$, inserting
immediately after a bag $B$ that contains $u$ a path decomposition of the paths
attached to $u$ where we have added $B$ to every bag.  \end{proof}

\subsection{Algorithm for Treewidth}

\begin{theorem}\label{thm_tw_DP}
 The $(p,q)$-d\eds\ problem can be solved in time $4^{2\tw^2}(4 (q+1)(p+1))^{2\tw} \cdot n^{O(1)}$ on graphs of treewidth at most~$\tw$.
\end{theorem}

The rest of this subsection is devoted to the proof of Theorem~\ref{thm_tw_DP}. 
Let $G=(V,E)$ the input graph and we are given a rooted nice
tree decomposition of $G$ with width $\tw$.  For each node $t$ of
the decomposition, let $B_t$ denote the corresponding bag and $V_t$
denote the set of vertices appearing in $B_t$ or one of the descendants of $t$. 
For a vertex set $X\subseteq V,$ we denote by $E(X)$ the set of arcs both of whose endpoints lie in $X$. 
For a set of arcs $D$, let $D^+$ (respectively, $D^-$) be the set of all heads (respectively, tails) of arcs in $D$.
For a function $f:X\rightarrow F$ and the subset $X'\subseteq X$, we denote the restriction of $f$ to $X'$ 
by $f|_{X'}$.

\bigskip

\noindent {\bf Feasible solution.} We construe a solution for the $(p,q)$-d\eds\ problem a triple $(D,f,b)$, 
where $D\subseteq E$, $f:V\rightarrow \{0,\ldots, q-1,\infty\}$ and $b:V\rightarrow \{0,\ldots , p-1,\infty\}$. Informally, the functions $f$ and $b$ keep track of the forward and backward distance from selected arcs.
A triple $(D,f, b)$ is said to be a \emph{feasible solution} (for the input instance $G$) if the following holds:
\begin{enumerate}[(i)]
\item for every arc $(x,y)\in E$, $f(x)<\infty$ or $b(y)<\infty$, 
\item for every $x\in V$ with $f(x)<\infty$, either $x\in D^+$ or there exists $x'\in \delta^-(x)$ with $f(x') < f(x)$,
\item for every $y\in V$ with $b(y)<\infty$, either $y\in D^-$ or there exists $y'\in \delta^+(y)$ with $b(y') <  b(y).$
\end{enumerate}

The \emph{size} of a feasible solution $(D,f,b)$ is defined as the cardinality
of $D$.  Here, the symbol $\infty$ represents a prohibitively large number. If
there is no path from a vertex $w$ to $x$, we set $\dist(w,x)=\infty.$ By
convention $\infty + c=\infty$ and $c<\infty$ hold for any finite number $c$.
If $q=0$ (resp. $p=0$), then the range of $f$ (resp. $b$) is simply
$\{\infty\}$. 

\begin{lemma}\label{claim:global}
$G$ allows a $(p,q)$-edge dominating set of size at most $d$ 
if and only if there is a feasible solution $(D,f,b)$ of size at most $d.$ 
\end{lemma}
\begin{proof}
Suppose that $D$ is a $(p,q)$-edge dominating set, and for every $x\in V$ let
$$f(x)=
\begin{cases}
\dist(D^+,x) \qquad &\text{if } \dist(D^+,x)<q\\
\infty \qquad &\text {otherwise.}
\end{cases}
$$ 
and 
$$b(x)=
\begin{cases}
\dist(x,D^-) \qquad &\text{if } \dist(x,D^-)<q\\
\infty \qquad &\text {otherwise.}
\end{cases}
$$ 
To see that $(D,f,b)$ is a feasible solution, first note that for every $(x,y)\in E$ 
either $\dist(D^+,x)<q$ or $\dist(y,D^-)<p$, hence $f(x)<\infty$ or $b(y)<\infty$ holds, i.e. the feasibility condition (i) holds. 
Without loss of generality, we assume the former. 
It remains to observe that either $\dist(D^+,x)=0$ or $x$ has an in-neighbor $x'$ on the shortest path from $D^+$ and $x$, thus (ii) holds.

Conversely, suppose that $(D,f,b)$ is a feasible solution. To see that $D$ is a $(p,q)$-dominating set, 
it suffices to show that $\dist(D^+,x)\leq f(x)$ and $\dist(x,D^-)\leq b(x)$ for every $x\in V$ because the feasibility condition (i) 
then implies that every arc $(x,y)$ is $(p,q)$-dominated by $D$. We prove the inequality $\dist(D^+,x)\leq f(x)$ by induction of the value $f(x)$; 
the  inequality  $\dist(x,D^-)\leq b(x)$  can be shown similarly. 
Note that if $f(x)=0$, the feasibility condition (ii) enforces that $x\in D^+$. Hence we have $\dist(D^+,x)\leq f(x)$ in this case. 
Assume that $\dist(D^+,x)\leq f(x)$ for every $x$ with $f(x)\leq i<\infty.$ 
If $i=q-1$, then we are done since the inequality trivially holds for all $z$ with $f(z)=\infty$. Therefore, we assume 
that $i+1<\infty$. Consider an arbitrary $z\in V$ with $f(z)=i+1.$ If $z\in D^+$, 
then clearly we have $0=\dist(D^+,z)\leq f(z).$
Otherwise, $z$ has an in-neighbor $z'$ with $f(z')< f(z)=i+1$ by (ii). By induction hypothesis, we conclude $f(z)\geq f(z')+1 \geq \dist(D^+,z')+1\geq \dist(D^+,z).$
\end{proof}

\bigskip

\noindent {\bf Partial solution, feasibility of a partial solution, witness.} The above formulation of a solution seems convoluted but 
it is useful for defining a partial solution for the dynamic programming algorithm. 
For a node $t$, a \emph{partial solution at $t$} is a triple $(D_t,f_t,b_t)$, where $D_t\subseteq E(V_t)$, 
$f_t:V_t \rightarrow \{0,\ldots , q-1, \infty\}$ and $b_t: V_t \rightarrow \{0,\ldots , p-1, \infty\}$. 
A partial solution $(D_t,f_t,b_t)$  is said to be \emph{feasible} at $t$
if  the following holds:
\begin{enumerate}[(a)]
\item for every arc $(x,y)\in E(V_t)$, $f_t(x)<\infty$ or $b_t(x)<\infty$, 
\item for every $x\in V_t$ with $f_t(x)<\infty$, either $x\in D_t^+$, or there exists $x'\in \delta^-(x)\cap V_t$ with $f_t(x')< f_t(x),$ or $x\in B_t,$
\item for every $y\in V_t$ with $b_t(y)<\infty$, either $y\in D_t^-$, or there exists $y'\in \delta^+(y)\cap V_t$ with $b_t(y')< b_t(y)$ or $y\in B_t.$
\end{enumerate}

We say that a vertex $x\in V_t$ has an \emph{$f$-witness} for a partial solution $(D_t,f_t,b_t)$ if 
$f_t(x)=\infty$, $x\in D_t^+$ or there exists $x'\in \delta^-(x)\cap V_t$ with $f_t(x') < f_t(x).$ 
In each case, $x$ itself, an arc $(u,v)\in D$ 
with $v=x$, and an in-neighbor $x'$ with $f_t(x') < f_t(x)$ is called an $f$-witness of $x$ for $(D_t,f_t,b_t)$. Likewise, 
$x$ itself when $b(x)=\infty$, an arc $(u,v)\in D_t$ with $x=u$, or an out-neighbor $x'$ in $V_t$ with $b_t(y')< b_t(y)$ 
is a \emph{$b$-witness} of $x$ for $(D_t,f_t,b_t)$. From the definition of witness and the feasibility conditions (a)-(c), 
the next observation is immediate.

\begin{fact}\label{fact:feasible}
A partial solution $(D_t,f_t,b_t)$ is feasible at $t$ if and only 
the feasibility condition (a) holds and every vertex $x\in V_t\setminus B_t$ has both an $f$-witness and a $b$-witness 
for $(D_t,f_t,b_t)$.
\end{fact}

\noindent {\bf Signature $\tau$, canonical signature, realizability.} We define a \emph{signature} $\tau$ at a node $t$ as a tuple consisting of the following entries.

\begin{itemize}
\item A set of arcs $A \subseteq E(B_t)$.
\item A non-negative integer $d$.
\item  $f:B_t\rightarrow \{0,\ldots , q\}$. 
\item  $b:B_t\rightarrow \{0,\ldots , p\}$. 
\item  $s_f:B_t \rightarrow \{0,1\}$.
\item  $s_b:B_t\rightarrow \{0,1\}$.
\end{itemize}

\smallskip

Intuitively speaking, a signature $\tau$ captures the projection of a feasible partial solution $(D_t,f_t,b_t)$ on $B_t$ 
and additionally keeps track of whether $x\in B_t$ has seen a witness or not with the indicator functions $s_f$ and $s_b$. 
The integer number $d$ intends to record the number of {\sl forgotten} arcs in the feasible partial solution.

Formally, a signature $\tau=(A,d,f,b,s_f,s_b)$ at $t$ is the \emph{canonical signature} of a feasible partial solution $(D_t,f_t,b_t)$ at $t$ if
\begin{enumerate}[(a)]
\item $A=D_t\cap E(B_t).$
\item $d=|D_t \setminus E(B_t)|$. 
\item $f=f_t|_{B_t}$
\item $b=b_t|_{B_t}$
\item for every $x\in B_t$, $s_f(x)=1$ if and only if $x$ has a $f$-witness for $(D_t,f_t,b_t)$.
\item for every $x\in B_t$, $s_b(x)=1$ if and only if $x$ has a $b$-witness for $(D_t,f_t,b_t)$.
\end{enumerate}
%
Notice that for each feasible partial solution at $t$ there is a unique canonical signature of it. 
A signature $\tau$ at node $t$ is \emph{realizable} if it is the canonical signature of a feasible partial solution at $t$.
We also remark that $s_f(x)=1$ for any $x\in B_t$ with $f(x)=\infty$   as $x$ itself is an $f$-witness of $x.$

\smallskip
The next claim is useful.

\begin{lemma}\label{claim:witness}
Let $\tau=(A,d,f,b,s_f,s_b)$ be a realizable signature at node $t$ and 
$(D_t,f_t,b_t)$ be a partial feasible solution which realizes $\tau$. 
Then, for each arc $(x,y)\in E(V_t)$ which is not $(p,q)$-dominated by $D_t$ in $G[V_t]$ 
there exists a vertex $w_0\in B_t$ such that 
one of the following holds.
\begin{itemize}
\item $f(w_0)+\dist(w_0,x)\leq q-1$ and $s_f(w_0)=0$, or
\item $\dist(y,w_0) +b(w_0)\leq p-1$ and $s_b(w_0)=0$.
\end{itemize}
\end{lemma}
\begin{proof}
If $D_t$ $(p,q)$-dominates every arc of $G[V_t]$, the claim trivially holds, so suppose this is not the case. 
Consider an arc $(x,y)\in E(V_t)$ which is not $(p,q)$-dominated by $D_t$ in $G[V_t]$. 
By the (partial) feasibility condition (a), we have $f_t(x)\leq q-1$ or $b_t(y)\leq p-1$. Without loss of generality, assume $f_t(x)\leq q-1.$ 
By the feasibility condition (b), there exists a sequence of vertices $(x=)x_0, x_1,\ldots , x_\ell$ in $V_t$ such that $q-1\geq f_t(x_0)>f_t(x_1)>\cdots >f_t(x_\ell)$ 
and $x_\ell, \ldots, x_1,x_0$ forms a directed path of $G[V_t]$; 
we choose a maximal such sequence. Because the value of $f_t$ strictly decreases along the sequence, we have $\ell\leq q-1$. This means 
that $x_i \notin D_t^+$ for every $i\in \{0,\ldots, \ell\}$ since otherwise, $D_t$ $(p,q)$-dominates $(x,y)$ in $G[V_t]$, contradicting the choice of $(x,y).$ 
By $x_\ell\notin D_t^+$, $f_t(x_\ell)<\infty$ and the maximality assumption  on the sequence, $x_\ell$ cannot have an $f$-witness for $(D_t,f_t,b_t)$, which implies $x_\ell\in B_t$ 
by Fact~\ref{fact:feasible}. 
In particular, the condition (e) of the canonical signature indicates that $s_f(x_\ell)=0.$ 
Lastly, observe that the construction of the sequence ensures that 
$q-1\geq f_t(x_0)\geq f_t(x_{i+1})+1 \geq \cdots \geq f_t(x_\ell)+\dist(x_\ell,x_0).$ 

The proof is symmetric when $b_t(y)\leq p-1$ holds instead.
\end{proof}

\begin{lemma}\label{claim:root}
There exists a $(p,q)$-edge dominating set of $G$ of size at most $d'$ if and only if there exists a realizable signature 
$\tau=(A_\tau,d_\tau, f^\tau,b^\tau,s_f^\tau,s_b^\tau)$ at the root node 
such that (i) $|A_\tau|+d_\tau \leq d',$ and (ii) $s_f^\tau(w)=s_b^\tau(w)=1$ for every $w\in B_t$.
\end{lemma}
\begin{proof}
Suppose that $D$ is a $(p,q)$-edge dominating set of $G$ of size at most $d'$. 
Let $(D,f^*,b^*)$ be a (global) feasible solution of size at most $d$; the existence of such a solution is guaranteed by Lemma~\ref{claim:global}. 
Thanks to the global feasibility condition (ii)-(iii), every vertex of $V$ has an $f$-witness as well as a $b$-witness. In particular this means that 
in the canonical signature $\tau=(A,d, f,b,s_f,s_b)$ of $(D,f^*,b^*)$, where $(D,f^*,b^*)$ is seen as a partial feasible solution at the root $r$, 
we have $s_f(x)=s_b(x)=1$ for every $x\in B_t.$ Clearly, $|A|+d\leq d'$ by the conditions (a)-(b) of the canonical signature. 
Therefore $\tau$ satisfies (i)-(ii) in the statement.

Conversely, let $\tau=(A_\tau,d_\tau, f^\tau,b^\tau,s_f^\tau,s_b^\tau)$ be a realizable signature at the root $r$ 
which meets the conditions (i)-(ii) of the statement. Let $(D_r,f_r,b_r)$ be a feasible solution whose canonical signature at $r$ is $\tau.$ We want to 
prove that $D_r$ is a $(p,q)$-edge dominating set of $G$ of size at most $d'.$
By the conditions (a)-(b) of the canonical signature, 
we have $|D_r|=|D_r\cap E(B_r)|+|D_r\setminus E(B_r)|=|A_\tau|+d_\tau$, which is at most $d'$ by the condition (i) of the statement. 

It remains to see that $D_r$ $(p,q)$-dominates every arc of $G$. Suppose not, and $(x,y)\in E=E(V_r)$ 
is not $(p,q)$-dominated by $D_r$. By Lemma~\ref{claim:witness}, there exists a vertex $w_0\in B_r$ with $s_f(w_0)=0$ or $s_b(w_0)=0$, 
which is impossible due to the condition (ii) in the statement.
\end{proof}

\smallskip

\noindent {\bf Computing all valid signatures.} For two  signatures $\tau$ and $\tau'$ at node $t$, we say that 
$\tau$ is \emph{superior} to $\tau'$ if all the entries of $\tau$ and $\tau'$ are identical except for the integer entry, in which $\tau$ takes a strictly smaller 
value than $\tau'$ does. A signature $\tau$ at $t$ is \emph{supreme} if there is no other realizable signatures at $t$ which is superior to $\tau.$ 
A signature is \emph{valid} if it is realizable and supreme.
Thanks to Lemma~\ref{claim:root}, it is sufficient to design a bottom-up procedure which produces all valid signatures at each node $t$ (and possibly some invalid ones as well) 
and determines whether a signature is valid or not, 
provided that all valid signatures have been computed for the children of $t$ (and invalid ones have been discarded). 

We provide such a procedure for each type of a tree node $t$: 
{\sl leaf, introduce, join} and {\sl forget} and argue that a signature $\tau$ at node $t$ is generated {\sl and} declared valid 
if and only if $\tau$ is indeed a valid signature. 

\bigskip

\noindent {\sl $\bullet$ Leaf node.} Let $B_t=\{w\}$. 
We generate all signatures $\tau=(A,d, f,b,s_f,s_b)$ with $A=\emptyset$, $d=0$, $f(w)\in \{0,\ldots , q-1,\infty\}$, $b(w)\in \{0,\ldots , p-1,\infty\}$ and 
$s_f(w)\in \{0,1\}, s_b(w)\in \{0,1\}$. 
We discard all unrealizable signatures. Deciding whether a signature is realizable or not is trivial in this case; 
simply check whether $f(w)=\infty$ if and only $s_f(w)=1$, likewise $b(w)=\infty$ if and only if $s_b(w)=1.$

\bigskip

\noindent {\sl $\bullet$ Introduce node.} Let $w$ be a newly introduced vertex and $B_t=B_{t'}\cup \{w\}.$
For a (not necessarily feasible) partial solution $(D_{t'},f_{t'},b_{t'})$ at node $t'$ and a triple $(\bar{A}, r,s ) \in 2^{\delta(w)}\times  \{0,\ldots , q-1,\infty\} \times \{0,\ldots , q-1,\infty\}$, 
the \emph{extension of $(D_{t'},f_{t'},b_{t'})$ by $(\bar{A}, a,b)$} is a partial solution $(D_t,f_t,b_t)$ at $t$ 
such that $D_t=D_{t'}\cup \bar{A}$, $f_t(x)=f_{t'}(x)$ for every $x\in B_{t'}$ and $f_t(w)=r$, and $b_t(x)=b_{t'}(x)$ for every $x\in B_{t'}$ and $b_t(w)=s$. 
We first observe that not only the extension of a partial solution is well-defined but also the extension of a signature by such a triple is well-defined. 

\begin{lemma}\label{claim:extequivalent} 
Let $(D_t,f_t,b_t)$ and $(D_{t'},f_{t'},b_{t'})$ be feasible partial solutions at node $t$ and $t'$ respectively, 
and let $(\bar{A}, a,b) \in  2^{\delta(w)}\times  \{0,\ldots , q-1,\infty\} \times \{0,\ldots , q-1,\infty\}$. 
Suppose $(D_t,f_t,b_t)$ is the extension of $(D_{t'},f_{t'},b_{t'})$ by the triple $(\bar{A}, r,s)$. 
Then 
\begin{enumerate}
\item Any vertex of $V_{t'}$ which has an $f$-witness for $(D_{t'},f_{t'},b_{t'})$ also has  an $f$-witness for $(D_t,f_t,b_t)$. 
\item $x\in B_{t'}$ does not have an $f$-witness for $(D_{t'},f_{t'},b_{t'})$ and has an $f$-witness  $(D_t,f_t,b_t)$ 
if and only if $x\in \delta^+(w)$, and either $x\in \bar{A}^+$ or $f(x)>f(w)$ holds. 
\item Any witness of $w$ for $(D_t,f_t,b_t)$ is in $G[B_t].$
\end{enumerate}
The symmetric statement holds for $b$-witnesses.
\end{lemma}
\begin{proof}
The first two statements are clear from the construction and the definition of $f$-witness. 
To see the third statement, it suffices to observe that any witness of $w$ is either $w$ itself, an arc incident with $w$ or an in-neighbor of $w$ in $V_t.$ 
In the first two cases, it is obvious that the witness is in $G[B_t]$. In the last case, note that $B_{t'}$ is a separator between $w$ and $V_t\setminus B_t$ 
and thus, a witness of $w$ as an in-neighbor of $w$ must be contained in $B_t$. 
\end{proof}

Lemma~\ref{claim:extequivalent} implies that if two feasible partial solutions at $t'$ have the same canonical signature $\tau'$ at $t'$, their extensions 
by a fixed triple $(\bar{A}, r,s)$ have the same canonical signature at $t$. 
This leads us to define the \emph{extension of a signature $\tau'$ at $t'$ by $(\bar{A}, r, s)$}.
For a signature $\tau'=(A',d', f',b',s'_f,s'_b)$ at node $t'$ and a triple $(\bar{A}, a,b) \in 2^{\delta(w)}\times  \{0,\ldots , q-1,\infty\} \times \{0,\ldots , q-1,\infty\}$, 
the \emph{extension of $\tau'$ by $(\bar{A}, r,s)$} is the signature $\tau=(A,d, f,b,s_f,s_b)$ at $t$ defined as
\begin{itemize}
\item $A=A'\cup \bar{A}$, 
\item $d=d'$, 
\item $f(x)=f'(x)$ for every $x\in B_{t'}$ and $f(w)=r$, 
\item $b(x)=b'(x)$ for every $x\in B_{t'}$ and $b(w)=s$, 
\item for every $x\in B_t$, $s_f(x)=1$ if and only if (i) $x\in \bar{A}^+$, or (ii) there exists $x'\in N^-(x)\cap B_t$ with $f(x') < f(x)$, or (iii) $f(x)=\infty$, or (iv) $x\in B_{t'}$ and $s'_f(x)=1,$ 
\item for every $y\in B_t$, $s_b(y)=1$ if and only if (i) $y\in \bar{A}^-$, or (ii) there exists $y'\in N^+(y)\cap B_t$ with $b(y') < b(y)$, or (iii) $b(y)=\infty$, or (iv) $y\in B_{t'}$ and $s'_b(y)=1,$
\end{itemize}

To obtain the set of all valid signatures at $t$, we consider all extensions over all valid signatures at $t'$ 
by all triple $(\bar{A}, a,b) \in 2^{\delta(w)}\times  \{0,\ldots , q-1,\infty\} \times \{0,\ldots , q-1,\infty\}$ such that 
the next two conditions are met.
\begin{itemize}
\item [$(*)$] for every arc $(w,x) \in E(B_t)\cap \delta^+(w)$, either $f(x)<\infty$ or $r<\infty$ holds, and 
\item [$(**)$] for every arc $(x,w) \in E(B_t)\cap \delta^-(w)$, either $b(x)<\infty$ or $s<\infty$ holds.
\end{itemize}
Then 
among the obtained signatures, supreme signatures are marked and the unmarked signatures are discarded. That this procedure 
produces all valid signatures follows from Lemma~\ref{lem:ext2}. Moreover, any generated signature is realizable by Lemma~\ref{lem:ext1}. 
Therefore, those signatures which are marked as supreme are precisely the set of all valid signatures at $t.$

\begin{lemma}\label{lem:ext2}
Let $\tau=(A,d, f,b,s_f,s_b)$ be a realizable signature at node $t$. Then there exists  a triple $(\bar{A}, r,s) \in 2^{\delta(w)}\times  \{0,\ldots , q-1,\infty\} \times \{0,\ldots , q-1,\infty\}$ and 
a realizable signature $\tau'=(A',d', f',b',s'_f,s'_b)$ at node $t'$ such that  $\tau$ is the extension of $\tau'$ by $(\bar{A}, r,s)$, 
and the conditions $(*)$ and $(**)$ hold.
\end{lemma}
\begin{proof}
Let $(D_\tau,f_\tau,b_\tau)$ be a partial feasible solution at $t$ whose canonical signature is $\tau$, 
and let $(D_\tau \setminus \delta(w), f_\tau|_{V_{t'}}, b_\tau|_{V_{t'}})$ be a partial solution at $t'$. It is clear that the latter is feasible at $t'$.  
Consider the triple $(A\cap \delta(w),f(w),b(w))$ and the canonical signature $\tau'=(A',d', f',b',s'_f,s'_b)$ of 
$(D_\tau \setminus \delta(w), f_\tau|_{V_{t'}}, b_\tau|_{V_{t'}})$ at $t'$. 
It is tedious to check that $\tau$ is the extension of $\tau'$ by $(A\cap \delta(w),f(w),b(w))$. 
The conditions $(*)$ and $(**)$ are met because $(D_\tau,f_\tau,b_\tau)$ is feasible, 
and due to the construction of the triple and $\tau'$.
\end{proof}

\begin{lemma}\label{lem:ext1}
Let $\tau'=(A',d', f',b',s'_f,s'_b)$ be a realizable signature at node $t'$. Then for any triple $(\bar{A}, a,b) \in 2^{\delta(w)}\times  \{0,\ldots , q-1,\infty\} \times \{0,\ldots , q-1,\infty\}$, 
the extension $\tau=(A,d, f,b,s_f,s_b)$ of $\tau'$ by $(\bar{A}, a,b)$ is realizable if and only if
the conditions $(*)$ and $(**)$ hold.
\end{lemma}
\begin{proof}
Let us see the `only if' part. If any of $(*)$ and $(**)$ fails to hold, then any partial solution whose canonical signature is $\tau$ fails to meet the feasibility condition (a), and thus 
cannot be a feasible partial solution at $t.$ 

For the `if' direction, let $(D_{\tau'},f_{\tau'},b_{\tau'})$ be a feasible partial solution which realizes $\tau'$ and let $(D_\tau,f_\tau,b_\tau)$ be the extension of it by the triple $(\bar{A}, a,b)$. 
It is tedious to check that $(D_\tau,f_\tau,b_\tau)$ is a feasible partial solution at $t$; the feasibility condition (a) 
is guaranteed by the feasibility of $(D_{\tau'},f_{\tau'},b_{\tau})$ and because  the conditions $(*)$ and $(**)$ 
hold for $(D_{\tau},f_{\tau},b_{\tau})$ and the triple $(\bar{A}, a,b)$. Also the feasibility condition (b) is
 satisfied due to the statement 1 of Lemma~\ref{claim:extequivalent}. 
 It remains to observe that $\tau$ is the canonical signature of $(D_\tau,f_\tau,b_\tau)$, which is again tedious to verify 
 using Lemma~\ref{claim:extequivalent}.
\end{proof}

\bigskip

\noindent {\sl $\bullet$ Join node.} Let $t_1$ and $t_2$ be the two children of $t$ with $B_t=B_{t_1}=B_{t_2}.$
For two signatures $\tau_i=(A^i,d^i, f^i,b^i,s^i_f,s^i_b)$ at node $t_i$ for $i=1,2$ which are \emph{compatible}, i.e. $A^1=A^2$, $f^1=f^2$ and $b^1=b^2$, 
the \emph{join} $\tau=(A, d, f,b,s_f,s_b)$ of $\tau_1$ and $\tau_2$ is defined as follows.  
\begin{itemize}
\item $A=A^1=A^2$.
\item $d=d_1+d_2.$
\item $f=f^1=f^2$, $b=b^1=b^2,$ 
\item for every $x\in B_t$, $s_f(x)=s^1_f(x) \vee s^2_f(x),$ and 
\item for every $x\in B_t$, $s_b(x)=s^1_b(x) \vee s^2_b(x)$.
\end{itemize}

For every compatible pair of valid signatures at $t_1$ and $t_2$, we generate the join. 
After that, we only keep the supreme signatures and discard the rest. That 
the signatures obtained in this way form the set of all valid signatures at node $t$ follows immediately 
from the next lemma.  

\begin{lemma} \label{lem:join}
A signature is valid at $t$ if and only if it is the join of two valid signatures at $t_1$ and $t_2$ which are compatible.
\end{lemma}
\begin{proof}
Let $\tau=(A, d, f,b,s_f,s_b)$ be a realizable signature at $t$ with 
a partial feasible solution $(D_\tau, f_\tau, b_\tau)$ realizing $\tau.$ 
Let $(D_i,f_i,b_i)$  be the partial solution at $t_i$ where $D_i=D_\tau\cap E(V_{t_i}),$ $f_i=f_\tau|_{V_{t_i}}$ and $b_i=b_\tau|_{V_{t_i}}$ for each $i=1,2.$ 
Clearly,  $(D_i,f_i,b_i)$ is feasible at $t_i$ for each $i=1,2.$
Let $\tau_i=(A^i,d^i, f^i,b^i,s^i_f,s^i_b)$ for $i=1,2$ be the canonical signature of $(D_i,f_i,b_i)$.
It is clear that  $\tau_1$ and $\tau_2$ are compatible and $\tau$ is the join of them. 

Conversely, let $\tau_i=(A^i,d^i, f^i,b^i,s^i_f,s^i_b)$ be realizable signatures at node $t_i$ for $i=1,2$ with $A^1=A^2$, $f^1=f^2$ and $b^1=b^2$ 
and let $(D_{\tau_i}, f_{\tau_i},b_{\tau_i})$ be a feasible partial solution realizing $\tau_i$ for $i=1,2.$ 
Let $D=D_{\tau_1}\cup D_{\tau_2}$, $f=f_{\tau_1}\cup f_{\tau_2}$ and $b=b_{\tau_1}\cup b_{\tau_2}$. 
We argue that the join $\tau=(A, d, f,b,s_f,s_b)$ of $\tau_1$ and $\tau_2$ is the canonical signature of $(D,f,b).$ 
The conditions (a)-(d) of a canonical signature is straightforward from $B_{t_1}=B_{t_2}=B_t$. To see (e) and (f), 
notice that $x\in B_t$ has an $f$-witness (resp. $b$-witness) for $(D,f,b)$ if and only if 
$x$ has an $f$-witness (resp. $b$-witness) for at least one of $(D_{\tau_i}, f_{\tau_i},b_{\tau_i})$ for $i=1,2.$ 
The latter holds precisely when $s_f(x)=1$ (resp. $s_b(x)=1$). 

Lastly, if there is a realizable signature superior to $\tau$, then one can obtain 
a realizable signature superior to $\tau_1$ or $\tau_2.$ Moreover, if any of $\tau_1$ and $\tau_2$ allows a realizable signature superior to it, 
one can obtain a realizable signature superior to $\tau.$ This completes the proof.
\end{proof}

\smallskip

\noindent {\sl $\bullet$ Forget node.} Let $w$ be the forgotten vertex and $B_t=B_{t'}\setminus \{w\}.$ 
For each valid signature $\tau'=(A',d', f',b',s'_f,s'_b)$ at node $t'$, let $\tau=(A,d, f,b,s_f,s_b)$ be the \emph{restriction} of $\tau'$ on $B_t$, namely 
\begin{itemize}
\item $A=A'\setminus \delta(w).$
\item $d=d' + A'\cap \delta(w).$
\item $f(x)=f'(x)$, $b(x)=b'(x)$, $s_f(x)=s'_f(x)$ and $s_b(x)=s'_b(x)$ for every $x\in B_t.$
\end{itemize}
The new signature $\tau$ is declared valid if and only if it is a restriction of some $\tau'$ at $t'$ with $s'_f(w)=s'_b(w)=1$. 
We claim that a signature $\tau$ at $t$ is valid 
if and only if it is generated and then declared valid. 

\begin{lemma}\label{lem:forget}
A signature $\tau$ at $t$ is valid if and only if it is the restriction of a valid signature $\tau'$ at $t'$ with $s'_f(w)=1$ and $s'_b(w)=1.$
\end{lemma}
\begin{proof}
Suppose $\tau'$ is a realizable signature at $t'$ with $s'_f(w)=s'_b(w)=1$ and the restriction of $\tau'$ on $B_t$ is $\tau$. 
Let $(D_{\tau'}, f_{\tau'}, b_{\tau'})$ be a feasible partial solution at $t'$ which realizes $\tau'$. We first argue that 
$(D_{\tau'}, f_{\tau'}, b_{\tau'})$ is a feasible partial solution at $t$. 
Clearly, $D_{\tau'}$ is fully contained in $E(V_t)=E(V_{t'})$ and the feasibility condition (a)  holds. 
To see (b), it is sufficient to verify that every vertex $z\in V_t\setminus B_t$ has an $f$-witness. This holds for every $z\neq w$ 
because $(D_{\tau'}, f_{\tau'}, b_{\tau'})$ is a feasible partial solution at $t'$. For $z=w$, that $s'_f(w)=1$ 
implies that $w$ has an $f$-witness for $(D_{\tau'}, f_{\tau'}, b_{\tau'})$ by the condition (e) of the canonical signature. 
The feasibility condition (c) can be similarly verified. 
It remains to observe that $\tau$ is the canonical signature of $(D_{\tau'}, f_{\tau'}, b_{\tau'})$ at $t$.

%

Conversely, suppose that $\tau$ is a realizable signature at $t$ and let $(D_{\tau}, f_{\tau}, b_{\tau})$ 
be a feasible partial solution at $t$ realizing $\tau$. 
Because $(D_{\tau}, f_{\tau}, b_{\tau})$ is feasible at $t$, every vertex $z\in V_t\setminus B_t$ has an $f$-witness (resp. $b$-witness) 
for $(D_{\tau}, f_{\tau}, b_{\tau})$. This implies that $(D_{\tau}, f_{\tau}, b_{\tau})$ is feasible at $t'$. 
Let $\tau'=(A',d', f',b',s'_f,s'_b)$ be the canonical signature of $(D_{\tau}, f_{\tau}, b_{\tau})$ at $t'$ 
Clearly, the restriction of $\tau'$ on $B_t$ equals $\tau$. That $s'_f(w)=s'_b(w)=1$ follows from the fact that $w\notin B_t$ and thus 
it has an $f$-witness for $(D_{\tau}, f_{\tau}, b_{\tau})$ due to the feasibility of $(D_{\tau}, f_{\tau}, b_{\tau})$.

To complete the proof, notice that the above constructions in both directions also establish that if $\tau$ at node $t$ is supreme if and only 
if it is the restriction of some supreme signature $\tau'$ at $t'$ with $s'_f(w)=1$ and $s'_b(w)=1.$
\end{proof}

By Lemmas~\ref{lem:ext2},~\ref{lem:ext1},~\ref{lem:join} and~\ref{lem:forget}, the procedures presented for introduce, join and forget nodes 
generate precisely the set of valid signatures at each node $t$. Finally, we can correctly decide if $G$ has a $(p,q)$-edge dominating set of 
size at most $d$ by examining the signatures at the root node thanks to Lemma~\ref{claim:root}.

\bigskip 

\noindent {\bf Running time.} 
At each node $t$, the number of possible signatures, except for the integer value $d$, generated from the child(ren) is at most 
$4^{\tw^2}\cdot (q+1)^{\tw}\cdot (p+1)^{\tw}\cdot 2^{\tw} \cdot 2^{\tw}$. Note that the signatures which are not supreme 
are generated amongst these options from the children of $t$ and discarded, and all in all at most $4^{\tw^2}(4 (q+1)(p+1))^\tw$ 
signatures are generated and examined. Examining each signature for checking the validity can be executed in $4^{2\tw^2}(4 (q+1)(p+1))^{2\tw}$  time. 
This yields the claimed running time, and completes the proof of Theorem~\ref{thm_tw_DP}.

\section{On Tournaments}\label{sec_tournaments}

A complete complexity classification for the problems $(p,q)$-d\eds\ is
presented in this section. For
$p=q=1$, the problem is NP-hard under a randomized reduction while being
amenable to an FPT algorithm and polynomial kernelization, due to the results of
Sections~\ref{subsec:fpt} and~\ref{sec_kernel}. The hardness reduction is given in
Subsection~\ref{subsec:tournament11}.  When
$p=2$ or $q=2$, the complexity status of $(p,q)$-d\eds\ is equivalent to \textsc{Dominating Set} on tournaments and is discussed in Subsection~\ref{subsec:tournament2}. In the remaining cases, when $p+q\leq 1$, or $\max\{p,q\}\geq 3$, 
while neither of them equals 2, the
problems turn out to be in P (Subsection~\ref{subsec:tournamentpoly}).

\subsection{Hard: when $p=q=1$}\label{subsec:tournament11} We present a
randomized reduction from \textsc{Independent Set} to $(1,1)$-d\eds. Our
reduction preserves the size of the instance up to polylogarithmic factors; as
a result it shows that $(1,1)$-d\eds\ does not admit a $2^{n^{1-\epsilon}}$
algorithm, under the randomized ETH. Furthermore, our reduction preserves the
optimal value, up to a factor of $(1-o(1))$; as a result, it shows that
$(1,1)$-d\eds\ is APX-hard under randomized reductions.

Before moving on, let us give a high-level overview of our reduction. The first
step is to reduce \textsc{Independent Set} on cubic graphs to the following intermediate problem 
called \textsc{Almost Induced Matching}, also known as \textsc{Maximum Dissociation Number} in the literature (\cite{Yannakakis81a,XiaoK17}).
A subgraph of $G$ induced on a vertex set $S\subseteq V$ is called an \emph{almost induced matching}, if 
every vertex $v\in S$ has degree $\le1$ in $G[S]$.

\begin{definition}
The problem \textsc{Almost Induced Matching (AIM)} takes as input an undirected graph $G=(V,E)$. The goal is 
to find an almost induced matching having the maximum number of vertices. 
\end{definition}

Our reduction creates an
instance of \textsc{Almost Induced Matching} that has several special
properties, notably producing a bipartite graph $G=(A,B,E)$. From this we then build our instance for $(1,1)$-d\eds. The basic strategy
will be to construct a tournament $T=(V',E')$, where $V'=A\cup B\cup C$
and $C$ is a set of new vertices. All edges of $E$ will be directed from $A$ to
$B$, non-edges of $E$ will be directed from $B$ to $A$, and all other edges
will be set randomly.  This intuitively encodes the structure of $G$ in $T$.

The idea is now that a solution $S$ in $G$ (that is, a set of vertices of $G$
that induces a graph with maximum degree $1$) will correspond to an edge
dominating set in $T$ where all vertices except those of $S$ will have total
degree $2$, and the vertices of $S$ will have total degree $1$ (in the solution). In particular,
vertices of $S\cap A$ will have out-degree $1$ and in-degree $0$, and vertices
of $S\cap B$ will have in-degree $1$ and out-degree~$0$.

The random structure of the remaining arcs of the tournament $T$ is useful in
two respects: in one direction, given the solution $S$ for $G$, it is easy to
deal with vertices that have degree $1$ in $G[S]$: we select the corresponding
arc from $A$ to $B$ in $T$. For vertices of degree $0$, however, we are forced
to look for edge-disjoint paths that will allow us to achieve our degree goals.
Such paths are guaranteed to exist if $C$ is random and large enough.  In the
other direction, given a good solution in $T$ we would like to guarantee that,
because the internal structure of $A$, $B$, and $C$ is chaotic, the only way to
obtain a large number of vertices with low degree is to place those with
in-degree $0$ in~$A$, and those with out-degree $0$ in $B$. The main result of this subsection is the following.

\begin{theorem}[Main]\label{thm_11_NP}
$(1,1)$-d\eds\ on tournaments cannot be solved in polynomial time, unless NP $\subseteq $ BPP. Furthermore, $(1,1)$-d\eds\ is APX-hard under randomized reductions, and does not admit an algorithm running in time $2^{n^{1-\epsilon}}$ for any $\epsilon$, unless the randomized ETH is false.
\end{theorem}

To prove Theorem~\ref{thm_11_NP}, we first reduce the \textsc{Independent Set} problem on cubic graphs to \textsc{Almost Induced Matching}. 
Before presenting the first reduction, we recall here the following theorem(s) for Independent Set, that will act as our starting point.

\begin{theorem}\cite{AlimontiK00,CyganFKLMPPS15}
\textsc{Independent Set} is APX-hard  on cubic graphs. Furthermore, \textsc{Independent Set} cannot be solved in time $2^{o(n)}$ unless the ETH is false.
\end{theorem}

Concerning  \textsc{Almost Induced Matching}, the problem is known to be NP-complete on bipartite graphs of maximum degree 3 and on $C_4$-free bipartite graphs~\cite{BoliacCL04}. 
It is also NP-hard to approximate on arbitrary graphs within a factor of $n^{1/2-\epsilon}$ for any  $\epsilon>0$~\cite{OrlovichDFGW11}. Our next lemma supplements the known hardness results 
on bipartite graphs and might be of independent interest.

\begin{figure}[htbp]
\centerline{\includegraphics[width=80mm]{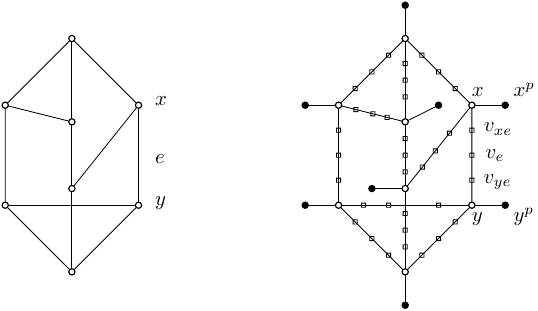}}
\caption{An example of our construction for Lemma~\ref{lem:AIMnphard}, with $G$ on the left and $G'$ on the right.}
\label{fig:AIM}
\end{figure}

\begin{lemma}\label{lem:AIMnphard}
\textsc{Almost Induced Matching} is APX-hard and cannot be solved in time
$2^{o(n)}$ under the ETH, even on bipartite graphs of degree at most 4.
Furthermore, this hardness still holds if we are promised that:
\begin{itemize}
 \item $OPT_{AIM}>0.6n$;
 \item there is an optimal solution $S$ that includes at least $n/20$ vertices
with degree $0$ in $G[S]$.
\end{itemize}\end{lemma} 

\begin{proof}
Let a graph $G=(V,E)$ and a positive integer $k$ be the input of \textsc{Independent Set}. 
We construct a graph $G'=(V',E')$ by subdividing each edge $e=(x,y)$ with three vertices $v_{xe},v_{e},v_{ye}$ so that 
the edge $e=(x,y)$ is replaced by a length-four path $x,v_{xe},v_{e},v_{ye},y$. 
In addition, we create a copy $x^p$ of each vertex $x\in V$ of $G$ and add it to $G'$ as a pendant vertex adjacent only to $x$ (see Figure~\ref{fig:AIM}). 
Fix $L=n+2m+k$. The vertices of $G'$ corresponding to the original vertices of $G$ are considered to inherit their labels in $G$ and we denote them as $V$. 
We prove that $G$ has an independent set of size $k$ if and only if $G'$ has 
an almost induced matching on $L$ vertices.

Suppose that $S$ is an independent set of $G$ with $|S|\geq k$. 
We construct a vertex set $S'$ of $G'$ so as to contain all vertices of $\{x^p: x\in V\}\cup S$ and also to 
include precisely one vertex set $\{v_{e},v_{ye}\}$ for each edge $e\in E$, where $y\notin S$. 
Since $S$ is an independent set, such a vertex set $S'$ exists. 
It is clear that $|S'|=n+k+2m$ and also that $G'[S']$ has degree at most one, meaning it is an almost induced matching of~$G'$. 
Conversely, let $S'$ be an almost induced matching of $G'$ of maximum size, and
suppose $|S'|\geq L$.  First observe that, without loss of generality,
we can assume that $S'$ contains all vertices of degree $1$. If a degree-one
vertex is not in $S'$ we add it, and remove its neighbor from~$S'$.

We now choose $S'$ so as to maximize the number of subdividing vertices
contained in $S'$.  We argue that for each edge $e=(x,y)\in E$, it holds that
$|S'\cap \{v_{xe},v_e,v_{ye}\}|=2$.  Clearly $|S'\cap
\{v_{xe},v_e,v_{ye}\}|\leq 2$.  Moreover, $S'$ contains at least one of
$\{v_{xe},v_e,v_{ye}\}$, since otherwise $S'\cup \{v_e\}$ is an almost induced
matching, contradicting the choice of $S'$. Suppose $|S'\cap
\{v_{xe},v_e,v_{ye}\}|=1$.  If $S'\cap \{v_{xe},v_e,v_{ye}\}=\{v_{xe}\}$, then
$v_{xe}$ must be matched with $x$ in $G'[S']$, as otherwise $S'\cup \{v_e\}$
is an almost induced matching. Then the set $S'\cup \{v_e\}\setminus \{x\}$
has strictly more subdividing vertices, giving rise to a contradiction.  Therefore, we have
$S'\cap \{v_{xe},v_e,v_{ye}\}=\{v_e\}$.  Now, the maximality of $S'$ implies
that both $x$ and $y$ are contained in $S'$.  Observe that $S'\cup
\{v_{xe}\}\setminus \{x\}$ is an almost induced matching of the same size as
$S'$ having strictly more subdividing vertices, producing a contradiction once more.  Therefore, we
have $|S'\cap \{v_{xe},v_e,v_{ye}\}|=2$ for every $e=(x,y)\in E$.

Moreover, this implies that for every $e=(x,y)\in E$, set $S'$ contains at most one of
$x$ and $y$, because, as $S'$ contains all leaves, if $x,y\in S'$, then
$v_{xe},v_{ye}\not\in S'$, which would mean that $S'$ only contains one of
$\{v_{xe},v_e,v_{ye}\}$.  Thus $S'\cap V$ corresponds to an independent set of
$G$.  It remains to note that $S'\cap (V\cup \{x^p:x\in V\})$ has at least $n+k$
vertices, and subsequently $S'\cap V$ has at least $k$ vertices.  This shows
that \textsc{Almost Induced Matching} is NP-hard. Observe also that the constructed
instance $G'$ is bipartite with one side of the bipartition including vertices $x^P,v_{xe},v_{ye},y^P$ and the other including vertices $x,v_e,y$ for every edge $e=(x,y)$ of $G$. 

To complete the proof, we note that when $G$ is a cubic graph, the constructed graph $G'$ has degree at most 4. 
Moreover, the hard instances of $G$ restricted to cubic graphs satisfy $k>n/4$, 
since any cubic graph on $n$ vertices has an independent set of size $\lceil n/4 \rceil.$ 
Now, it is straightforward to verify that the above reduction is an
$L$-reduction (i.e., linear) from \textsc{Independent set} on cubic graphs to \textsc{Almost
Induced Matching} on bipartite graphs of degree at most 4. The APX-hardness of
the former establishes the APX-hardness of the latter. Furthermore, the number
of vertices of the new graphs is linear in $n$. The inequality noted above for $k$ gives our properties' desired bounds.
\end{proof}

As our construction is randomized, the following (technical) property of a uniform random
tournament will be useful. Intuitively, the property established in Lemma
\ref{lem:nostrongbias1} below states that it is impossible in a large random
tournament to have two large sets of vertices $X,Y$, such that all vertices of
$X$ have in-degree $0$ and out-degree $1$ in a $(1,1)$-edge dominating set,
while all vertices of $Y$ have in-degree $1$ and out-degree~$0$.

\begin{lemma}\label{lem:nostrongbias1}
Let $T=(V,E)$ be a random tournament on the vertex set $\{1,2,\dots,n\}$, in which $(i,j)$ is an arc of $T$ with probability $1/2$. 
Then the following event happens with high probability: 
for any two disjoint  sets $X,Y\subseteq V$ with $|X|>(\log n)^2$ and $|Y|>
(\log n)^2$, there exists a vertex $x\in X$ with at least two outgoing arcs to
$Y$ and a vertex $y\in Y$ with at least two incoming arcs from $X$.
\end{lemma}
\begin{proof}
Fix arbitrary sets $X$ and $Y$ satisfying the stated cardinality conditions. We
will show that the claimed vertex $x$ exists with high probability and the
proof is symmetric for $y$. 

Let $|X|=s_1>\log^2n$ and $|Y|=s_2>\log^2n$.  We say that $(X,Y)$ is
\emph{strongly biased} if each $x\in X$ has at most one outgoing arc to $Y$.
Then we have:

\begin{align*} 
{\sf Prob}[(X,Y) \text{ is strongly biased}] &\leq \big (2^{-s_2} \cdot s_2\big )^{s_1}\\ 
&\le 2^{-s_1s_2 + 2(\log n)^3}\le 2^{-\frac{s_1s_2}{2}},  
\end{align*} 
\noindent where the last inequality follows from the lower bounds on $s_1,s_2$. Applying the union bound, the probability that $T$ has a strongly biased pair
$(X,Y)$ with $|X|=s_1, |Y|=s_2$
is at most 

\[2^{-\frac{s_1s_2}{2}} \cdot n^{ s_1}n^{s_2}\leq 2^{-\frac{s_1s_2}{4}}, \]

\noindent for any sufficiently large
$n$. This probability is smaller than $\frac{1}{n^3}$ for sufficiently large $n$ and thus taking the union bound over all possible values of $s_1,s_2$ gives the claim. 
\end{proof}

Another useful (albeit also technical) property of the random digraphs we will be employing in our construction, concerning the existence of vertex-disjoint directed paths, is given next. 

\begin{lemma}\label{lem:nostrongbias3}

Let $G=(V=A \dot{\cup} B \dot{\cup} C,E)$ be a random directed graph with
$|A|=|B|=n$ and $|C|=4n$, such that for any pair $(x,y)$ with $\{x,y\}\cap C\neq
\emptyset$ we have exactly one arc, oriented from $x$ to $y$, or from $y$ to
$x$ with probability $1/2$. Let $\ell\geq n/20$ be a positive integer.  Then
with high probability, we have: for any two disjoint  sets $X\subseteq
A$, $Y\subseteq B$ with $|X|=|Y|=\ell$, there exist $\ell$ vertex-disjoint
directed paths from $X$ to $Y$.

\end{lemma} 
\begin{proof} Suppose that there do not exist $\ell$ vertex-disjoint directed
paths from $X$ to $Y$ and let $T\subseteq X\cup C\cup Y$ be a  minimal
$(X,Y)$-separator of size at most $\ell-1$. We have $|C\setminus T|\ge 3n+1$.
We say that a vertex $u\in C\setminus T$ is \emph{helpful}, if there exists $v_1\in X$
and $v_2\in Y$ such that $(v_1,u),(u,v_2)$ are arcs of the graph. Clearly, if
$T$ is a separator, $C\setminus T$ must not contain any helpful vertices.

A vertex $u\in C$ is not helpful if either all edges between $u$ and $X$ are
oriented towards $X$, or all arcs between $u$ and $Y$ are oriented towards $u$.
Each of these events happens with probability at most $2^{-n/20}$. Therefore,
the probability that all the vertices of $C\setminus T$ (being at least $3n+1$) are not
helpful is at most $2^{-\frac{3n^2}{20}}$ (as these events are independent).
This is an upper-bound on the probability that two specific sets $X,Y$ do not
have $|X|$ vertex disjoint sets connecting them, and are therefore separated by
a set $T$.  Taking the sum over all the
choices for $X,Y,T$ (being at most $2^n\cdot 2^n\cdot 2^{4n}$) and using the union bound, we conclude that no such sets exist with high probability (as $n$ increases).
\end{proof}

We are now ready to present our construction in Theorem~\ref{prop:11edsreduction} below. Our construction is randomized and rather technical, making use of the specific properties held by the intermediate instances produced by the above transformation from Independent Set (Lemma~\ref{lem:AIMnphard}).

\begin{theorem}[Construction]\label{prop:11edsreduction}

Suppose we are given an instance of \textsc{Almost Induced Matching} on a
bipartite graph with $2n$ vertices and maximum degree $4$ such that there is an
optimal solution that induces at least $n/10$ vertices of degree $0$. There is
a randomized algorithm which runs in time polynomial in $n$ and, given an
integer $L\ge 1.2n$, reduces the \textsc{Almost Induced Matching} instance to an
instance $T$ of $(1,1)$-d\eds, such that $T$ is a tournament with $O(n)$
vertices and we also have with high probability:

\begin{itemize}

\item[(a)] if $OPT_{AIM}(G)\geq  L$, then $OPT_{(1,1)dEDS}(T)\leq |V(T)|-L/2
+1$;

\item[(b)] if $OPT_{AIM}(G)<  L - 5(\log L)^2$, then $OPT_{(1,1)dEDS}(T)>
|V(T)|-L/2 +1$.

\end{itemize}
\end{theorem}
\begin{proof} 
Let $G=(A\dot{\cup}B,E)$ be an input bipartite graph of
\textsc{Almost Induced Matching} with maximum degree at most 4. We may assume
that no vertex of $G$ is isolated.  We may also assume that $|A|=|B|=n$, and if
$S$ is an almost induced matching of $G$ with $|S|\geq L$ then $|S\cap
A|=|S\cap B|$, by taking the disjoint union of two copies of $G$.  This means
that we may also assume that $L$ is even.  

From $G$, we construct a tournament $T$ on the vertex set $A'\dot{\cup}
B'\dot{\cup} C$, where $A'=\{x':x'\in A\}$, $B'=\{x':x' \in B\}$ and $|C|=4n$.
The arc set of $T$ is formed as follows (see Figure~\ref{fig:rand_tour}):

\begin{itemize}
\item for every pair of vertices $x\in A$ and $y\in B$, $(x,y)\in A(T)$, if and only if $(x,y)\in E$.
\item $T[A']$, $T[B']$, $T[C]$ are random tournaments in which each pair $u,v$ of vertices gets an orientation $u\rightarrow v$ with probability $1/2$, independently.
\item For every $a\in A'$ and $c\in C$, we have an orientation $a\rightarrow c$ with probability $1/2$, independently. The same holds between $B'$ and $C$.
\end{itemize}

\begin{figure}[htbp]
\centerline{\includegraphics[width=60mm]{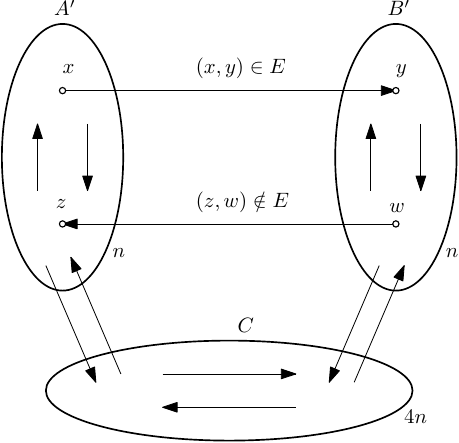}}
\caption{A simplified representation of our construction for Theorem~\ref{prop:11edsreduction}.}
\label{fig:rand_tour}
\end{figure}

We first prove (a): Suppose that $S$ is an almost induced matching containing at
least $L$ vertices, and let $S_0$ and $S_1\subseteq S$ be the sets of all
vertices having degree exactly $0$ and $1$ in $G[S]$, respectively. Slightly abusing
notation, let $S_0$ and $S_1$ refer to the corresponding vertex sets in $T$.  Note
that $|S_0\cap A'|=|S_0\cap B'|\geq n/20$.  We construct an  arc set $D$ of $T$
as follows.  Let $M$ be the set of arcs defined as $\delta(S_1\cap A',S_1\cap
B')$. We include all arcs of $M$ in $D$. 

By Lemma \ref{lem:nostrongbias3}, there exist (with high probability) $|S_0\cap A|$
vertex-disjoint directed paths $\mathcal{P}$ from $S_0\cap A$ to $S_0\cap B$.
We add to $D$ all arcs contained in a path of $\mathcal{P}$, denoted as
$E(\mathcal{P})$.  

Let us now observe that, with high probability, $T$ does not contain any sources
or sinks, as the probability that a vertex is a source or a sink is at most
$2^{-n}$, and there are $O(n)$ vertices in $T$. We use this fact to complete
the solution as follows: consider the digraph $T'=T-S_1-V(\mathcal{P})$, where
$V(\mathcal{P})$ is the set of all vertices contained in a path of
$\mathcal{P}$.  Recall that any tournament has a Hamiltonian path that can be found in polynomial time.  We choose a
directed Hamiltonian path $Q$ of $T'$, with $s$ and $t$ as the start and end
vertices of $Q$.  We add all the arcs $E(Q)$ of $Q$ to $D$, plus one incoming
arc $(s',s)$ of $s$ and one outgoing arc $(t,t')$ of $t$. Since we have no
sources or sinks, such arcs $(s',s)$ and $(t,t')$ exist. Note that $|D'|\leq
|V(T')|+1$.

We argue that the obtained arc set 
\[D=E(M)\cup E({\mathcal P}) \cup E(Q) \cup \{(s',s),(t,t')\}\] is a
$(1,1)$-edge dominating set of $T$. First note that all internal vertices of
the disjoint paths $\mathcal{P}$, as well as all vertices of $T'$ have both
positive in-degree and positive out-degree, therefore all arcs incident on such
vertices are covered. For edges induced by $S_0\cup S_1$, we have that all arcs
of this type going from $A$ to $B$ have been selected (since $S$ is an almost
matching), and all arcs going in the other direction are covered as all
vertices of $(S_0\cup S_1)\cap A$ have positive out-degree.

Lastly, we observe 
\begin{align*}
|D|&=|V(M)|-|S_1|/2+|V(\mathcal{P})|-|S_0|/2 + (|V(T)|-|V(M)|-|V(\mathcal{P})|
+1)\\ &\leq  |V(T)|-L/2+1.  \end{align*}

To see (b), let $D$ be a $(1,1)$-edge dominating set of $T$ of size at most
$|V(T)|-L/2 + 1$. We will use this to build a large almost induced matching in
$G$.  We define the following vertex sets:

\begin{align*}
R_{0,pos}&=\{v\in V(T): d_D^-(v)=0\quad \text{and} \quad d_D^+(v)>0\}\\
R_{0,1}&=\{v\in V(T): d_D^-(v)=0\quad \text{and} \quad d_D^+(v)=1\}\\
R_{pos,0}&=\{v\in V(T): d_D^-(v)>0 \quad \text{and} \quad d_D^+(v)=0\}\\
R_{1,0}&=\{v\in V(T): d_D^-(v)= 1 \quad \text{and} \quad d_D^+(v)=0\}
\end{align*}

Clearly, it holds that $R_{0,1}\subseteq R_{0,pos}$ and $R_{1,0}\subseteq R_{pos,0}$. 
By definition, the arc set from $R_{0,pos}$ to $R_{pos,0}$ must be completely contained in $D$, since no such arc can be 
$(0,1)$-dominated or $(1,0)$-dominated, and the arc is thus required to dominate itself. 
\begin{equation}
\delta(R_{0,pos},R_{pos,0})\subseteq D \label{eqn:takeall}
\end{equation}

Given this, we observe that $(R_{0,1}\cap A' ) \cup (R_{1,0}\cap B')$, seen as a vertex set of $G$ sharing the same vertex names, 
is an almost induced matching of $G$. If that is not so, then either there exists $x\in R_{0,1}\cap A'$ with two outgoing arcs to $R_{1,0}\cap B'$, or
$y\in R_{1,0}\cap B'$ with two incoming arcs from $R_{0,1}\cap A'$. In the
former case, both outgoing arcs from $x$ must be contained in $D$ as previously
noted. This means $x\notin R_{0,1}$, however, which gives a contradiction. A
symmetric argument holds in the latter case.

Our aim is then to show that a ``good chunk'' of $R_{0,1}$ is contained in $A'$, and that of $R_{1,0}$ in $B'$. We will use the following claim.

\begin{claim}\label{claim:inequalities} We have $|R_{0,pos}|\geq L/2-1$,
$|R_{pos,0}|\geq L/2-1$ and $|R_{0,1}| +|R_{1,0}|\geq L-4$.  \end{claim}

\begin{proof}
Consider the numbers $\sum_{v\in V(T)} d_D^-(v)$ and $\sum_{v\in V(T)}
d_D^+(v)$, where $d_D^-(v), d_D^+(v)$ denote the number of arcs of $D$ going
into and coming out of $v$, respectively.  As every arc $(x,y)\in D$ is counted
precisely once in each sum, it holds that

\[|D|=\sum_{v\in V(T)} d_D^-(v)= \sum_{v\in V(T)} d_D^+(v).\]

We now have

\begin{align*} |V(T)|-L/2+1 &\geq |D| = \sum_{v\in V(T)} d^-_D(v) = \sum_{i}
i\cdot |\{v\in V(T):d^-_D(v)=i\}|\\ &\geq |V(T)|- |R_{0,pos}|, \end{align*} 

\noindent from which it follows that $|R_{0,pos}|\geq L/2-1$ and similarly
$|R_{pos,0}|\geq L/2-1$. Also, observe that there is at most one vertex $v$
with $d_D(v)=0$, where $d_D(v)$ is the total number of arcs of $D$ incident on
$v$.  Indeed, if there are two such vertices $u$ and $v$ then the arc between
$u$ and $v$ cannot be $(1,1)$-dominated. We therefore have:

\begin{align*} 2|V(T)|-L +2 &\geq 2|D| = \sum_{v\in V(T)} d_D(v) = \sum_{i}
i\cdot |\{v\in V(T):d_D(v)=i\}|\\ &\geq |R_{0,1}| +|R_{1,0}| + 2 (|V(T)|-
|R_{0,1}| - |R_{1,0}| -1) \end{align*} 

\noindent establishing the claimed inequalities.  
\end{proof}

We can now resume the proof of Theorem~\ref{prop:11edsreduction} (reduction).
By (\ref{eqn:takeall}) and the definition of $R_{0,1}$,  every $x\in R_{0,1}$
has at most one outgoing arc to $R_{pos,0}$, because as we previously argued,
all such arcs are included in $D$. Consider now the bigger of the three sets
among $R_{pos,0}\cap A'$, $R_{pos,0}\cap B'$ and $R_{pos,0}\cap C$. The biggest
of these sets must have size at least $L/6$ which is larger than $(\log n)^2$
for sufficiently large $n$. We apply Lemma \ref{lem:nostrongbias1} on
$R_{0,1}\cap C$ and the largest of the three aforementioned sets. We conclude
that $|R_{0,1}\cap C|\le (\log n)^2$, because otherwise there is a vertex in
$R_{0,1}\cap C$ which has two outgoing arcs to $R_{pos,0}$, which is a
contradiction. With symmetric arguments for $R_{1,0}\cap C$ we have
\begin{equation} |R_{0,1}\cap C|\leq (\log n)^2 \qquad \text{and} \qquad
|R_{1,0}\cap C|\leq (\log n)^2.  \label{eqn:smallInC} \end{equation} That is,
most vertices of $R_{0,1}$ and $R_{1,0}$ can be found in $A'\cup B'$.

We now concentrate on the four sets $R_{0,1}\cap A'$, $R_{1,0}\cap A'$,
$R_{0,1}\cap B'$ and $R_{1,0}\cap B'$. We will say that one of these sets is
``large'' if its cardinality is at least $(\log n)^2$.  The following claim more
carefully specifies which combinations of these sets may be simultaneously large.

\begin{claim}\label{large_claim}
Precisely two of the following sets have size larger than $(\log n)^2$:  $R_{0,1}\cap A'$, $R_{1,0}\cap A'$, $R_{0,1}\cap B'$, $R_{1,0}\cap B'$.
Furthermore, it holds that:
\begin{itemize}
\item either $|R_{0,1}\cap A'|>(\log n)^2$ and $|R_{1,0}\cap B'|>(\log n)^2$,
\item or $|R_{1,0}\cap A'|>(\log n)^2$ and $|R_{0,1}\cap B'|>(\log n)^2$.
\end{itemize} 
\end{claim} 
\begin{proof}
Because, from Claim \ref{claim:inequalities} we have $|R_{0,1}|+|R_{1,0}|\ge
L-4$ and $L\ge 1.2n$, if we take into account that $|A'|=|B'|=n$ and the fact
that $|R_{0,1}\cap C|$ and $|R_{1,0}\cap C|$ are at most $(\log n)^2$, we
conclude that at least two of the four sets we focus on ($R_{0,1}\cap A',
R_{0,1}\cap B', R_{1,0}\cap A', R_{1,0}\cap B'$) must be large, that is, have
cardinality at least $(\log n)^2$.  

We now propose the following facts: (i) if $R_{0,1}\cap A'$ is large, then
only $R_{1,0}\cap B'$ is large; (ii) if $R_{0,1}\cap B'$ is large, then only
$R_{1,0}\cap A'$ is large; (iii) $R_{1,0}\cap A'$ and $R_{1,0}\cap B'$
cannot be simultaneously large. It is not hard to see that these three
statements together give the claim.

To see (i) suppose that $|R_{0,1}\cap A'|$ is large. We argue that
$|R_{pos,0}\cap A'|\le (\log n)^2$. Indeed, if not, then by Lemma
\ref{lem:nostrongbias1} there exists a vertex in $R_{0,1}\cap A'$ which has two
outgoing arcs to $R_{pos,0}\cap A'$, a contradiction. Therefore, $|R_{1,0}\cap
A'|\le (\log n)^2$. Furthermore, we must have $|R_{pos,0}\cap C| \le (\log
n)^2$. Indeed, otherwise we again invoke Lemma \ref{lem:nostrongbias1} to find
a vertex in $R_{0,1}\cap A'$ with two outgoing arcs to $R_{pos,0}\cap C$, a
contradiction. Since by Claim \ref{claim:inequalities} we have that
$|R_{pos,0}|\ge L/2-1$ it must be the case that $|R_{pos,0}\cap B'|\ge(\log
n)^2$. If we have $|R_{0,1}\cap B'|\ge (\log n)^2$ then by Lemma
\ref{lem:nostrongbias1} we have a vertex in $R_{0,1}\cap B'$ with two outgoing
arcs to $R_{pos,0}\cap B'$, a contradiction. Therefore, $|R_{0,1}\cap B'|$ is
also small, and hence the only other set that may be large is $R_{1,0}\cap B'$.

To see (ii) it suffices to see that this statement is symmetric to (i) with the
roles of $A',B'$ reversed, so identical arguments apply.

Finally, to see (iii), suppose that $|R_{1,0}\cap A'|, |R_{1,0}\cap B'|\ge
(\log n)^2$. We argue that $|R_{0,pos}\cap A'|\le (\log n)^2$. Indeed,
otherwise by Lemma \ref{lem:nostrongbias1} we have a vertex $y\in R_{1,0}\cap
A'$ with two incoming arcs from $R_{0,pos}\cap A'$, a contradiction. With a
similar argument $|R_{0,pos}\cap B'| \le (\log n)^2$. Therefore, since
$|R_{0,pos}|\ge L/2-1$ by Claim \ref{claim:inequalities}, we must have
$|R_{0,pos}\cap C|\ge (\log n)^2$. This also gives rise to a contradiction, however,
since we can apply Lemma \ref{lem:nostrongbias1} to find a vertex $y\in
R_{1,0}\cap A'$ with two incoming arcs from $R_{0,pos}\cap C$.  \end{proof}

We can now complete the proof of our reduction, Theorem~\ref{prop:11edsreduction}.
Suppose that the first case of Claim~\ref{large_claim} above holds, meaning $|R_{1,0}\cap
A'|>(\log n)^2$ and $|R_{0,1}\cap B'|>(\log n)^2$.  For every $x \in B'$, we
know that the in-degree of $x$ with respect to $A'$ is at most 4 because we
reduce from an input instance $G$ whose degree is at most 4. Therefore, $x\in
R_{0,1}\cap B'$ has at least $(\log n)^2-4$ outgoing arcs to $R_{1,0}\cap A'$.
All such arcs must be included in $D$ by (\ref{eqn:takeall}), however, contradicting the definition of $R_{0,1}$.  Therefore, we have:

\begin{align*}
|R_{0,1}\cap A'|>(\log n)^2 \qquad &\text{and} \qquad |R_{1,0}\cap B'|>(\log n)^2\\
|R_{1,0}\cap A'|\leq (\log n)^2 \qquad &\text{and} \qquad |R_{0,1}\cap B'|\leq (\log n)^2.
\end{align*}
With Inequalities (\ref{eqn:smallInC}) and Claim~\ref{claim:inequalities}, we get:
\[|R_{0,1}\cap A'| +|R_{1,0}\cap B'|\geq |R_{0,1}| + |R_{1,0}|-4(\log n)^2 \geq L-4-4(\log n)^2.\]
Therefore $(R_{0,1}\cap A') \cup (R_{1,0}\cap B')$, seen as a vertex subset of $G$, is an almost induced matching of size at least $L-4-4(\log n)^2$. 
From $n\leq 2L$, we establish property (b) of the theorem's statement for sufficiently large $n$.
\end{proof}

\begin{proof}[Proof of Theorem \ref{thm_11_NP} (Main)]
Let $G$ be an instance of \textsc{Independent Set} on cubic graphs and let $G'$ be the instance of \textsc{Almost Induced Matching} 
obtained by the construction of Lemma~\ref{lem:AIMnphard}. We set $\ell$ as in the reduction and observe that 
$OPT_{IS}(G)\geq  k$, if and only if $OPT_{AIM}(G')\geq \ell$.

Let $G^*$ be a disjoint union of $10(\log \ell)^2$ copies of $G'$. Then $G^*$ is a gap-instance, 
whose optimal solution is of size either at least $10\ell(\log \ell)^2$, or at most $10\ell(\log \ell)^2- 10(\log \ell)^2 \leq L - 5(\log L)^2$, 
where $L:=10\ell(\log \ell)^2$. Now, Theorem~\ref{prop:11edsreduction}
implies that using a probabilistic polynomial-time algorithm for $(1,1)$-d\eds\
with two-sided bounded errors, one can correctly decide an instance of
\textsc{Independent Set} on cubic graphs with bounded errors. We observe that
the size of the instance has only increased by a poly-logarithmic factor, hence
an algorithm solving the new instance in time $2^{n^{1-\epsilon}}$ would give a
randomized sub-exponential time algorithm for \textsc{3-SAT}.

Finally, for APX-hardness, we observe that we may assume we start our reduction
from an \textsc{Independent Set} instance where either $OPT_{IS}\ge k$, or
$OPT_{IS}< rk$, for some constant $r<1$ and $k=\Theta(n)$. Lemma
\ref{lem:AIMnphard} then gives an instance of \textsc{Almost Induced Matching}
where either $OPT_{AIM}\ge L_1$, or $OPT_{AIM} \le r'L_1 = L_2$, for some
(other) constant $r'<1$. We now use Theorem \ref{prop:11edsreduction} to
create a gap-instance of $(1,1)$-d\eds. \end{proof}

\subsection{Equivalent to Dominating Set on tournaments:
$p=2$ or $q=2$}\label{subsec:tournament2}
We next consider the versions for $p=2$ or $q=2$ and show that they are W[2]-hard, while being solvable in $n^{O(\log n)}$. We begin with a series of lemmas that we then use to obtain the main theorems of this subsection.

\begin{lemma}\label{lem_02_DS}
On tournaments without a source, 
we have $OPT_{(0,2)dEDS}\le OPT_{DS}$. 
\end{lemma}
\begin{proof}
 Let $T=(V,E)$ be a tournament with no source and $D\subseteq V$ be a dominating set of $T$. Then let $K\subseteq E$ be a set containing one arbitrary incoming arc of every vertex in $D$. We claim $K$ $(0,2)$-dominates all arcs in $E$: since $D$ is a dominating set, for any vertex $u\notin D$ there must be an arc $(v,u)$ from some $v\in D$. Thus all outgoing arcs $(u,w)$ from such $u\notin D$ are $(0,2)$-dominated by $K$, as are all arcs $(v,u)$ from $v\in D$.
\end{proof}

\begin{lemma}\label{lem_wsource_r}
Let $T=(V,E)$ be a tournament and let $s$ be a source of $T$. Then $\delta^+(s)$ is an optimal $(p,q)$-edge dominating set of $T$, for any $p\le 1$ and  $q\geq 1$.
\end{lemma}
\begin{proof}
Since $s$ has no incoming arcs, any $(p,q)$-edge dominating set must select at least one arc from $\{(s,v)\} \cup \delta^+(v)$ for every $v\in V\setminus \{s\}$ in order to $(p,q)$-dominate 
$(s,v)$. Because the arc sets $\{(s,v)\} \cup \delta^+(v)$ are mutually disjoint over all $v\in V\setminus \{s\}$, any $(p,q)$-edge dominating set has size at least $|\delta^+(s)|$. 
Now, observe that $\delta^+(s)$  $(0,1)$-dominates every arc of $T$.
\end{proof}

\begin{lemma}\label{lem_22_EDSsize}
On tournaments on $n$ vertices, for any $p\geq 2$, it is $OPT_{(p,2)dEDS}\leq OPT_{(2,2)dEDS}\le 2 \log n+3$.
\end{lemma}
\begin{proof}
The first inequality trivially holds, so we prove the second inequality. 
Let $T=(V,E)$ be a tournament on $n$ vertices. If $T$ has no source, then $OPT_{(2,2)dEDS}\leq OPT_{(0,2)dEDS} \leq OPT_{DS}\leq \log n +1$, 
where the second and the last inequality follow from Lemma~\ref{lem_02_DS} and Lemma~\ref{lem_ds_tour}, respectively. 
If $T^{rev}$ contains no source, observe that a $(0,2)$-edge dominating set of $T^{rev}$ is a $(2,0)$-edge dominating set of $T$ and the statement holds. 

Therefore, we may assume that $T$ has a source $s$ and a sink $t$.  
Let $S_1\subseteq V\setminus \{s\}$ be a dominating set of $T-s$ of size at most $\log n +1$. 
Clearly, every arc $(u,v)$ of $T-s$ lies on a directed path of length at most two from some vertex of $S_1$. Let $D_1\subseteq E$ be a minimal arc set such that $D_1\cap \delta^-(v)\neq\emptyset$ for every $v\in S_1$. 
Since every $v\in S_1$ has positive in-degree, such a set $D_1$ exists and we have $|D_1|\leq |S_1|$. 
Observe that $D_1$ (0,2)-dominates every arc of $T-s$. Applying a symmetric argument to $T^{rev}-t$, we know that there exists an arc set $D_2$ of size at most $\log n+1$ 
which $(2,0)$-dominates every arc of $T-t$. 
Now $D_1\cup D_2$ (2,2)-dominates every arc incident with $V\setminus \{s,t\}$. Therefore, $D_1\cup D_2 \cup \{(s,t)\}$ is a  $(2,2)$-edge dominating set.
\end{proof}

\begin{lemma}\label{lem:with2EDS}
There is an FPT reduction from \textsc{Dominating Set} on tournaments parameterized by solution size to $(p,q)$-d\eds\ parameterized by solution size, when $p=2$ or $q=2$.
\end{lemma}
\begin{proof}
We assume that $q=2$, without loss of generality. Let $T=(V,E)$ be an input tournament to \textsc{Dominating Set}, and let $k$ be the solution size. 
It can be assumed that $T$ has no source.  We construct a tournament $T'$ by
adding to $T$ a new vertex $t$ which is a sink, meaning we orient all arcs from
$V$ to $t$. We claim that $OPT_{(p,2)dEDS}(T') = OPT_{DS}(T)$.

Given a dominating set $D$ of $T$, we select an arbitrary arc set $K$ of $T'$ so that $\delta_K^-(v)=1$ for each $v\in D$. 
It is easy to see that $K$ $(0,2)$-dominates every arc of $T'$: any arc $(u,v)$ with $u\in D$ 
is clearly dominated by $K$. For any arc $(u,v)$ with $u\notin D$, there is $w\in D$ such that $(w,u)\in E$ 
and thus $K$ $(0,2)$-dominates $(u,v)$. 

Conversely, suppose that $K$ is a $(p,2)$-edge dominating set of size at most
$k$ and let $K^+$ be the set of heads of $K$ found in $V$. Let $K^-$ be the set
of vertices $u\in V$ such that $(u,t)\in K$. We have $|K^+\cup K^-|\le k$,
because each arc of $K$ either contributes an element in $K^+$ or in $K^-$. We
claim that $K^+\cup K^-$ is a dominating set of $T$. Suppose the contrary,
therefore there exists $u\in V\setminus (K^+\cup K^-)$ that is not dominated by
$K^+\cup K^-$. The arc $(u,t)$, however, is dominated by $K$. We have
$(u,t)\not\in K$, as $u\not\in K^-$. Therefore, since $t$ is a sink, $(u,t)$ is
$(0,2)$-dominated by an arc $(v,w)\in K$. This means that either $w=u$, or the
arc $(w,u)$ exists. It is $w\in K^+$, however, meaning that $u$ is dominated.
 \end{proof}

\begin{theorem}\label{thm_02_12_22_domsetCom}
 On tournaments, the problems $(p,2)$-d\eds\ are W[2]-hard for each fixed $p$.
\end{theorem}
\begin{proof}
 For all problems, we use the reduction from \textsc{Set Cover} to \textsc{Dominating set on Tournaments} given by \cite{CyganFKLMPPS15} in Theorem~13.14 therein and our results follow from the W[2]-hardness of that problem (see also Theorem~13.28 therein) and our Lemma~\ref{lem:with2EDS} above.
\end{proof}

\begin{theorem}\label{thm_02_12_22_domsetALG}
 On tournaments, the problems $(0,2)$-d\eds, $(1,2)$-d\eds\ and $(p,2)$-d\eds, for any $p\ge2$, can be solved in time $n^{O(\log n)}$.
\end{theorem}
\begin{proof}
For $(0,2)$-d\eds\ and $(1,2)$-d\eds, the case when a given tournament contains a source can be solved in polynomial time by Lemma~\ref{lem_wsource_r}. If the input tournament contains no source, then 
by Lemma~\ref{lem_02_DS} we have $OPT_{(1,2)dEDS}\leq OPT_{(0,2)dEDS}\leq OPT_{DS}$, which is bounded by $\log n +1$ by Lemma~\ref{lem_ds_tour}.  
Lemma~\ref{lem_22_EDSsize} states that $OPT_{(p,2)dEDS} \leq 2\log n +3$. Exhaustive search over vertex subsets of size $O(\log n)$ performs in the claimed runtime.
\end{proof}

\subsection{P-time solvable: $p+q\leq 1$ or, $2\notin \{p,q\}$ and $\max\{p,q\}\geq 3$}\label{subsec:tournamentpoly}
We turn our attention to the remaining cases and show that they are in fact solvable in polynomial time.

\begin{theorem}\label{thm_01_P}

 $(0,1)$-d\eds\ can be solved  in polynomial time on tournaments.

\end{theorem}

\begin{proof}
 We will show that $OPT_{(0,1)dEDS}=n-1$ and give a polynomial-time algorithm for finding such an optimal solution. First, given a tournament $T=(V,E)$, to see why $OPT_{(0,1)dEDS}\ge n-1$ consider any optimal solution $K\subseteq E$: if there exists a pair of vertices $u,v\in V$ with $d_K^-(u)=d_K^-(v)=0$, meaning a pair of vertices, neither of which has an arc of $K$ as an incoming arc, then the arc between them (without loss of generality, let its direction be $(v,u)$) is not dominated: as $d_K^-(u)=0$, the arc itself does not belong in $K$ and as $d_K^-(v)=0$, there is no arc preceding it that is in $K$. This leaves 
 $(v,u)$ undominated. Therefore, there cannot be two vertices with no incoming arcs in any optimal solution, implying any  solution must include at least $n-1$ arcs.
 
To see $OPT_{(0,1)dEDS}\le n-1$, consider a partition of $T$ into strongly connected components $C_1,\dots,C_l$, where we can assume these are given according to their topological ordering, meaning for $1\le i<j\le l$, all arcs between $C_i$ and $C_j$ are directed towards $C_j$. Let $S$ be the set of arcs traversed in breadth-first-search (BFS) from some vertex $s\in C_1$ until all vertices of $C_1$ are spanned. Also let $S'$ be the set of arcs $(s,u),\forall u\in C_i,\forall i\in[2,l]$, that is, all outgoing arcs from $s$ to every vertex of $C_2,\dots,C_l$. Note that set $S'$ must contain an arc from $s$ to \emph{every} vertex that is not in $C_1$: $T$ being a tournament means every pair of vertices has an arc between them and $C_1$ being the first component in the topological ordering means all arcs between its vertices and those of subsequent components are oriented away from $C_1$. Then $K\coloneqq S\cup S'$ is a directed $(0,1)$-edge dominating set of size $n-1$ in $T$: observe that $d_K^-(u)=1, \forall u\not=s\in T$, that is, every vertex in $T$ has positive in-degree within $K$ except $s$. Thus all outgoing arcs from all such vertices $u$ are $(0,1)$-dominated by $K$, while all outgoing arcs from $s$ are in $K$, due to the BFS selection for $S$ and the definition of $S'$.
 
Since such an optimal solution $K$ can be computed in polynomial time
(partition into strongly connected components, BFS), the claim follows.
\end{proof}

\begin{theorem}\label{thm_pq3_P} For any $p,q$ with $\max\{p,q\}\ge3$, $p\neq 2$ and $q\neq 2$, $(p,q)$-d\eds\ can be solved in polynomial time on tournaments.
\end{theorem} \begin{proof}
Suppose, without loss of generality, that $q\geq 3$, as otherwise we can solve
$(q,p)$-d\eds\ on $T^{rev}$, the tournament obtained by reversing the
orientation of every arc. In any tournament $T$, there always exists a
\emph{king} vertex, that is, a vertex with a path of length at most 2 to any
other vertex in the graph. One such vertex is the vertex of maximum out-degree
$v$. If $v$ is not a source, it suffices to select one of its incoming arcs:
since there is a path of length at most 2 from $v$ to any other vertex $u$ in
the graph, any outgoing arc from any such $u$ will be $(0,3)$-dominated by this
selection. This is clearly optimal. 

Suppose now that $s$ is a source. We consider two cases: if $p\le 1$, then  Lemma~\ref{lem_wsource_r} 
implies that $\delta^+(s)$ is 
optimal.
Finally, suppose $s$ is a source and $p\ge 3$. If $T$ does not have a sink, then a king of $T^{rev}$ has 
an incoming arc, which $(0,3)$-dominates $T^{rev}$ as observed above, and thus $T$ 
has a $(3,0)$-edge dominating set of size~1. 

Therefore, we may assume that $T$ has both a source $s$ and 
a sink $t$. Let $s'$ and $t'$ be vertices of $V\setminus\{s,t\}$ with maximum out- and in-degree, respectively. 
Now $\{(s,t), (s,s'), (t',t)\}$ is a $(3,3)$-edge dominating set. This is because $s'$ is a king of $T-s$ and thus 
every arc $(u,v)$ with $u\neq s$ is $(0,3)$-dominated by $(s,s')$. Similarly, every arc $(u,v)$ with $v\neq t$ 
is $(3,0)$-dominated by $(t',t)$. The only arc not $(3,3)$-dominated by these two arcs is $(s,t)$, which is only
dominated by itself. Note this also implies optimality as any $(3,3)$-edge dominating set contains at least three arcs. Examining all vertex subsets of size up to 3, we can compute an optimal $(3,3)$-edge dominating set in polynomial time.
\end{proof}



\bibliography{deds}

\begin{thebibliography}{39}
\providecommand{\natexlab}[1]{#1}
\providecommand{\url}[1]{\texttt{#1}}
\expandafter\ifx\csname urlstyle\endcsname\relax
  \providecommand{\doi}[1]{doi: #1}\else
  \providecommand{\doi}{doi: \begingroup \urlstyle{rm}\Url}\fi

\bibitem[Alimonti and Kann(2000)]{AlimontiK00}
P.~Alimonti and V.~Kann.
\newblock Some {APX}-completeness results for cubic graphs.
\newblock \emph{Theoretical Computer Science}, 237\penalty0 (1-2):\penalty0
  123--134, 2000.

\bibitem[Baker(1994)]{Baker94}
B.~S. Baker.
\newblock {Approximation Algorithms for {NP}-Complete Problems on Planar
  Graphs.}
\newblock \emph{Journal of the ACM}, 41\penalty0 (1):\penalty0 153--180, 1994.

\bibitem[Belmonte et~al.(2018)Belmonte, Hanaka, Katsikarelis, Kim, and
  Lampis]{BelmonteHK0L18}
R.~Belmonte, T.~Hanaka, I.~Katsikarelis, E.~J. Kim, and M.~Lampis.
\newblock New results on directed edge dominating set.
\newblock In \emph{43rd International Symposium on Mathematical Foundations of
  Computer Science, {MFCS}}, pages 67:1--67:16, 2018.

\bibitem[Binkele{-}Raible and Fernau(2010)]{Binkele-RaibleF10}
D.~Binkele{-}Raible and H.~Fernau.
\newblock Enumerate and measure: Improving parameter budget management.
\newblock In \emph{{IPEC}}, volume 6478 of \emph{Lecture Notes in Computer
  Science}, pages 38--49. Springer, 2010.

\bibitem[Biswas et~al.(2022)Biswas, Jayapaul, Raman, and Satti]{BiswasJRS22}
A.~Biswas, V.~Jayapaul, V.~Raman, and S.~R. Satti.
\newblock Finding kings in tournaments.
\newblock \emph{Discret. Appl. Math.}, 322:\penalty0 240--252, 2022.

\bibitem[Bodlaender et~al.(2016)Bodlaender, Drange, Dregi, Fomin, Lokshtanov,
  and Pilipczuk]{BodlaenderDDFLP16}
H.~L. Bodlaender, P.~G. Drange, M.~S. Dregi, F.~V. Fomin, D.~Lokshtanov, and
  M.~Pilipczuk.
\newblock A ${O}(c^k n)$ 5-approximation algorithm for treewidth.
\newblock \emph{{SIAM} Journal of Computing}, 45\penalty0 (2):\penalty0
  317--378, 2016.

\bibitem[Boliac et~al.(2004)Boliac, Cameron, and Lozin]{BoliacCL04}
R.~Boliac, K.~Cameron, and V.~V. Lozin.
\newblock On computing the dissociation number and the induced matching number
  of bipartite graphs.
\newblock \emph{Ars Combinatorica}, 72, 2004.

\bibitem[Borradaile and Le(2016)]{BorradaileL16}
G.~Borradaile and H.~Le.
\newblock Optimal dynamic program for r-domination problems over tree
  decompositions.
\newblock In \emph{{IPEC}}, volume~63 of \emph{LIPIcs}, pages 8:1--8:23.
  Schloss Dagstuhl - Leibniz-Zentrum f{\"{u}}r Informatik, 2016.

\bibitem[Cardinal et~al.(2009)Cardinal, Langerman, and Levy]{CardinalLL09}
J.~Cardinal, S.~Langerman, and E.~Levy.
\newblock Improved approximation bounds for edge dominating set in dense
  graphs.
\newblock \emph{Theoretical Computer Science}, 410\penalty0 (8-10):\penalty0
  949--957, 2009.

\bibitem[Chleb{\'{i}}k and Chleb{\'{i}}kov{\'{a}}(2006)]{Chlebik2006}
M.~Chleb{\'{i}}k and J.~Chleb{\'{i}}kov{\'{a}}.
\newblock {Approximation hardness of edge dominating set problems}.
\newblock \emph{Journal of Combinatorial Optimization}, 11\penalty0
  (3):\penalty0 279--290, 2006.

\bibitem[Cygan et~al.(2015)Cygan, Fomin, Kowalik, Lokshtanov, Marx, Pilipczuk,
  Pilipczuk, and Saurabh]{CyganFKLMPPS15}
M.~Cygan, F.~V. Fomin, L.~Kowalik, D.~Lokshtanov, D.~Marx, M.~Pilipczuk,
  M.~Pilipczuk, and S.~Saurabh.
\newblock \emph{Parameterized Algorithms}.
\newblock Springer, 2015.

\bibitem[Demaine et~al.(2005)Demaine, Fomin, Hajiaghayi, and
  Thilikos]{DemaineFHT05}
E.~D. Demaine, F.~V. Fomin, M.~T. Hajiaghayi, and D.~M. Thilikos.
\newblock Fixed-parameter algorithms for $(k, r)$-center in planar graphs and
  map graphs.
\newblock \emph{{ACM} Transactions on Algorithms}, 1\penalty0 (1):\penalty0
  33--47, 2005.

\bibitem[Dinur and Steurer(2014)]{DinurS14}
I.~Dinur and D.~Steurer.
\newblock Analytical approach to parallel repetition.
\newblock In D.~B. Shmoys, editor, \emph{Symposium on Theory of Computing,
  {STOC}}, pages 624--633. {ACM}, 2014.

\bibitem[Downey and Fellows(1995{\natexlab{a}})]{Downey1995}
R.~G. Downey and M.~R. Fellows.
\newblock Fixed-parameter tractability and completeness i: Basic results.
\newblock \emph{SIAM Journal of Computing}, 24\penalty0 (4):\penalty0 873--921,
  1995{\natexlab{a}}.

\bibitem[Downey and Fellows(1995{\natexlab{b}})]{Downey95}
R.~G. Downey and M.~R. Fellows.
\newblock {Parameterized Computational Feasibility}.
\newblock In \emph{{Feasible Mathematics II}}, pages 219--244,
  1995{\natexlab{b}}.

\bibitem[Eisenstat et~al.(2014)Eisenstat, Klein, and Mathieu]{EisenstatKM14}
D.~Eisenstat, P.~N. Klein, and C.~Mathieu.
\newblock Approximating \emph{k}-center in planar graphs.
\newblock In \emph{Symposium on Discreet Algorithms {SODA} 2014}, pages
  617--627. {SIAM}, 2014.

\bibitem[Fellows et~al.(2009)Fellows, Hermelin, Rosamond, and
  Vialette]{FellowsHRV09}
M.~R. Fellows, D.~Hermelin, F.~A. Rosamond, and S.~Vialette.
\newblock On the parameterized complexity of multiple-interval graph problems.
\newblock \emph{Theoretical Computer Science}, 410\penalty0 (1):\penalty0
  53--61, 2009.

\bibitem[Fernau(2006)]{Fernau06}
H.~Fernau.
\newblock {Edge dominating set: Efficient Enumeration-Based Exact Algorithms}.
\newblock In \emph{{Parameterized and Exact Computation}}, pages 142--153.
  Springer Berlin Heidelberg, 2006.

\bibitem[Fomin et~al.(2009)Fomin, Gaspers, Saurabh, and Stepanov]{FominGSS09}
F.~V. Fomin, S.~Gaspers, S.~Saurabh, and A.~A. Stepanov.
\newblock On two techniques of combining branching and treewidth.
\newblock \emph{Algorithmica}, 54\penalty0 (2):\penalty0 181--207, 2009.

\bibitem[Fujito and Nagamochi(2002)]{FujitoN02}
T.~Fujito and H.~Nagamochi.
\newblock {A 2-approximation algorithm for the minimum weight edge dominating
  set problem.}
\newblock \emph{Discrete Applied Mathematics}, 118\penalty0 (3):\penalty0
  199--207, 2002.

\bibitem[Garey and Johnson(1979)]{GJ76}
M.~R. Garey and D.~S. Johnson.
\newblock \emph{Computers and Intractability: A Guide to the Theory of
  NP-Completeness}.
\newblock W. H. Freeman and Co., 1979.

\bibitem[Hagerup(2012)]{Hagerup12}
T.~Hagerup.
\newblock Kernels for edge dominating set: Simpler or smaller.
\newblock In \emph{Mathematical Foundations of Computer Science {MFCS} 2012},
  volume 7464 of \emph{Lecture Notes in Computer Science}, pages 491--502.
  Springer, 2012.

\bibitem[Hanaka et~al.(2019)Hanaka, Nishimura, and Ono]{HanakaNO17}
T.~Hanaka, N.~Nishimura, and H.~Ono.
\newblock On directed covering and domination problems.
\newblock \emph{Discrete Applied Mathematics}, 259:\penalty0 76 -- 99, 2019.

\bibitem[Harary and Norman(1960)]{Harary60}
F.~Harary and R.~Z. Norman.
\newblock Some properties of line digraphs.
\newblock \emph{Rendiconti del Circolo Matematico di Palermo}, 9\penalty0
  (2):\penalty0 161--168, 1960.

\bibitem[Horton and Kilakos(1993)]{Horton1993}
J.~D. Horton and K.~Kilakos.
\newblock Minimum edge dominating sets.
\newblock \emph{SIAM Journal on Discrete Mathematics}, 6\penalty0 (3):\penalty0
  375--387, 1993.

\bibitem[Iwaide and Nagamochi(2016)]{IwaideN16}
K.~Iwaide and H.~Nagamochi.
\newblock An improved algorithm for parameterized edge dominating set problem.
\newblock \emph{Journal of Graph Algorithms and Applications}, 20\penalty0
  (1):\penalty0 23--58, 2016.

\bibitem[Katsikarelis et~al.(2019)Katsikarelis, Lampis, and
  Paschos]{KatsikarelisLP17}
I.~Katsikarelis, M.~Lampis, and V.~T. Paschos.
\newblock Structural parameters, tight bounds, and approximation for
  $(k,r)$-center.
\newblock \emph{Discrete Applied Mathematics}, 264:\penalty0 90 -- 117, 2019.
\newblock Combinatorial Optimization: between Practice and Theory.

\bibitem[Kloks(1994)]{Kloks94}
T.~Kloks.
\newblock \emph{Treewidth, Computations and Approximations}, volume 842 of
  \emph{LNCS}.
\newblock Springer, 1994.

\bibitem[Kreutzer and Tazari(2012)]{Kreutzer2012}
S.~Kreutzer and S.~Tazari.
\newblock Directed nowhere dense classes of graphs.
\newblock In \emph{Symposium on Discrete Algorithms {SODA} 2012}, pages
  1552--1562, 2012.

\bibitem[Moshkovitz(2015)]{Moshkovitz15}
D.~Moshkovitz.
\newblock The projection games conjecture and the {NP}-hardness of $\ln
  n$-approximating set-cover.
\newblock \emph{Theory of Computing}, 11:\penalty0 221--235, 2015.

\bibitem[Orlovich et~al.(2011)Orlovich, Dolgui, Finke, Gordon, and
  Werner]{OrlovichDFGW11}
Y.~L. Orlovich, A.~Dolgui, G.~Finke, V.~S. Gordon, and F.~Werner.
\newblock The complexity of dissociation set problems in graphs.
\newblock \emph{Discrete Applied Mathematics}, 159\penalty0 (13):\penalty0
  1352--1366, 2011.

\bibitem[Schmied and Viehmann(2012)]{SchmiedV12}
R.~Schmied and C.~Viehmann.
\newblock Approximating edge dominating set in dense graphs.
\newblock \emph{Theoretical Computer Science}, 414\penalty0 (1):\penalty0
  92--99, 2012.

\bibitem[Vazirani(2001)]{V01}
V.~V. Vazirani.
\newblock \emph{{Approximation algorithms}}.
\newblock Springer, 2001.

\bibitem[Williamson and Shmoys(2011)]{Williamson2011}
D.~P. Williamson and D.~B. Shmoys.
\newblock \emph{The Design of Approximation Algorithms}.
\newblock Cambridge University Press, 1st edition, 2011.

\bibitem[Xiao and Kou(2017)]{XiaoK17}
M.~Xiao and S.~Kou.
\newblock Exact algorithms for the maximum dissociation set and minimum
  $3$-path vertex cover problems.
\newblock \emph{Theoretical Computer Science}, 657:\penalty0 86--97, 2017.

\bibitem[Xiao and Nagamochi(2014)]{XiaoN14}
M.~Xiao and H.~Nagamochi.
\newblock A refined exact algorithm for edge dominating set.
\newblock \emph{Theoretical Computer Science}, 560:\penalty0 207--216, 2014.

\bibitem[Xiao et~al.(2013)Xiao, Kloks, and Poon]{XiaoKP13}
M.~Xiao, T.~Kloks, and S.~Poon.
\newblock New parameterized algorithms for the edge dominating set problem.
\newblock \emph{Theoretical Computer Science}, 511:\penalty0 147--158, 2013.

\bibitem[Yannakakis(1981)]{Yannakakis81a}
M.~Yannakakis.
\newblock Node-deletion problems on bipartite graphs.
\newblock \emph{{SIAM} Journal Computing}, 10\penalty0 (2):\penalty0 310--327,
  1981.

\bibitem[Yannakakis and Gavril(1980)]{Yannakakis80}
M.~Yannakakis and F.~Gavril.
\newblock Edge dominating sets in graphs.
\newblock \emph{SIAM Journal on Applied Mathematics}, 38\penalty0 (3):\penalty0
  364--372, 1980.

\end{thebibliography}

\end{document}